\documentclass[sigconf]{acmart}
\acmConference[ArXiv]{}{December 2025}{Online}%
\usepackage[ruled,linesnumbered,vlined]{algorithm2e}
\usepackage{tikz}
\usepackage{hyperref}
\usepackage{listings}
\usepackage{enumitem}
\usepackage{subfigure}
\usepackage{color}
\usepackage{pgfplots}
\usepackage{multirow}
\usetikzlibrary{patterns}
\usepackage{balance}
\usepackage{subfiles}
\usepackage{tabularx}
\usepackage{pifont}
\usepackage{adjustbox}
\usepackage{mdframed}
\usepackage{makecell}
\usepackage{pifont}
\usepackage{utfsym}
\usetikzlibrary{patterns}

\usepackage{amssymb}
\usepackage{amsthm}
\usepackage{graphicx}
\usepackage{bbding}
\newcommand{\circnum}[1]{\ding{\numexpr171+#1\relax}}
\usepackage[inkscapeformat=png]{svg}

\newtheorem{theorem}{Theorem} 
\newtheorem{lemma}{Lemma}
\AtBeginDocument{%
  \providecommand\BibTeX{{%
    \normalfont B\kern-0.5em{\scshape i\kern-0.25em b}\kern-0.8em\TeX}}}

\newcounter{siqiang}
\numberwithin{siqiang}{section}

\newcommand{\cmark}{\Checkmark}

\newcommand\halfcmark{\cmark\kern-1.3ex\raisebox{1.0ex}{\rotatebox[origin=c]{125}{\textbf{---}}}}

\setcopyright{none}
\acmBooktitle{}
\acmISBN{}
\acmDOI{}
\acmPrice{}
\acmYear{2025}

\begin{document}
\title{RadixGraph: A Fast, Space-Optimized Data Structure for Dynamic Graph Storage (Extended Version)}

\author{Haoxuan Xie}
\email{haoxuan001@e.ntu.edu.sg}
\affiliation{%
  \institution{Nanyang Technological University}
  \country{Singapore}
}

\author{Junfeng Liu}
\email{junfeng001@e.ntu.edu.sg}
\affiliation{%
  \institution{Nanyang Technological University}
  \country{Singapore}
}

\author{Siqiang Luo}
\authornote{Corresponding author.}
\email{siqiang.luo@ntu.edu.sg}
\affiliation{%
  \institution{Nanyang Technological University}
  \country{Singapore}
}

\author{Kai Wang}
\authornote{Work done when the author was working as a research fellow at NTU.}
\email{kai\_wang@hit.edu.cn}
\affiliation{%
  \institution{Harbin Institute of Technology}
  \country{China}
}


\renewcommand{\shortauthors}{Haoxuan Xie, Junfeng Liu, Siqiang Luo, and Kai Wang}
\begin{abstract}
Dynamic graphs model many real-world applications, and as their sizes grow, efficiently storing and updating them becomes 
critical. We present RadixGraph, a fast and memory-efficient data structure for dynamic graph storage. RadixGraph features a carefully designed radix-tree-based vertex index that strikes an optimal trade-off between query efficiency and space among all pointer-array-based radix trees. For edge storage, it employs a hybrid snapshot-log architecture that enables amortized $O(1)$ update time. RadixGraph supports millions of concurrent updates per second while maintaining competitive performance for graph analytics. Experimental results show that RadixGraph outperforms the most performant baseline by up to $16.27\times$ across various datasets in ingesting graph updates, and reduces memory usage by an average of $40.1\%$. RadixGraph is open-source at \url{https://github.com/ForwardStar/RadixGraph}.
\end{abstract}

\maketitle
\section{Introduction}
\label{sec:intro}
Graphs are fundamental structures in computer science and widely used to model relationships and interactions between entities, including biological networks \cite{Ma2023,Forster2022,10.1093/bib/bbaa257}, financial networks \cite{cheng2025,10.1145/3641857,10.1145/3677052.3698648} and social networks \cite{10.1145/3539597.3570433,10.1145/3626772.3657962,10.1145/3336191.3371829}. As real-world graphs evolve continuously, the storage and processing of dynamic graphs have become a significant area of focus for both academia and industry. For example, Facebook leverages dynamic graphs to model user relationships at a trillion-scale \cite{10.14778/2824032.2824077}, while Twitter employs large-scale dynamic graphs to generate real-time recommendations \cite{10.14778/2733004.2733010}. These applications demand not only the management of large-scale graphs but also the efficient processing of updates to the evolving graph. Given the high computational cost of graph algorithms and the need for fast query efficiency, there is an urgent need to design in-memory dynamic graph data structures supporting fast query and update operations, while being space efficient~\cite{10.1145/3685980.3685984}. 
However, designing a graph data structure that supports dynamic updates is much more challenging than developing a static structure. It requires wise designs on both vertex index (using which we can locate an index) and edge index (using which we can store and retrieve neighbor edges of a vertex) to deliver both satisfying query and update performance. We summarize existing vertex and edge indices in Table \ref{tab:dyn-graph} and discuss their limitations as follows. 
\begin{table}[t]
\centering
\setlength{\tabcolsep}{2mm}{
\begin{tabular}{l|ccc}
\hline
\textbf{Vertex index}       & \textbf{Update} & \textbf{Query}  & \textbf{Space}           \\
\hline
Static array \cite{graphone,livegraph}      &   Fast     &   Fast     &   High              \\
Multi-level vector \cite{sortledton,gtx} &  Moderate      &  Fast      &  Moderate                \\
Tree-based \cite{cpam,teseo,aspen}       &  Moderate      &  Moderate      &  Low               \\
\hline
\textbf{Edge index}         & \textbf{Insert} & \textbf{Delete} & \textbf{Scan} \\
\hline
PMA-based \cite{pma,vcsr,10.1145/1292609.1292616}          &    $O(lg^2(d))$    &   $O(lg^2(d))$     &   Fast              \\
Log-based \cite{gtx,livegraph,graphone}         &   $O(1)$     &    $O(d)$    &     Fast           \\
Tree-based \cite{sortledton,gastcoco}         &   $O(lg(d))$     &   $O(lg(d))$    &   Moderate              \\
\hline
\end{tabular}}
\caption{Comparison of existing vertex and edge indices. ``Scan'' refers to the efficiency of sequential scan, which requires scanning a set of edges and is a common access pattern in graph queries. $d$ represents the average degree of vertices.}
\label{tab:dyn-graph}
\vspace{-1mm}
\end{table}

\vspace{1mm}
\noindent\textbf{Challenges of vertex index.} Existing graph systems adopt various vertex indexing schemes, each reflecting a trade-off among \textbf{update efficiency}, \textbf{query performance}, and \textbf{space utilization}. Common designs include: (1) \emph{static arrays} that directly index vertices of contiguous IDs, (2) \emph{multi-level vectors} that are often coupled with hash tables for ID translation, and (3) \emph{tree-based structures} such as B-trees~\cite{Btree} and adaptive radix trees (ART)~\cite{ART}.
\emph{Static arrays} \cite{graphone,livegraph} provide constant-time access but incur relatively high space overhead especially when vertex identifiers are non-contiguous. Unfortunately, in dynamic graphs, IDs are often hashed or generated as UUIDs (e.g., DataStax Enterprise Graph uses URIs\footnote{\url{https://docs.datastax.com/en/dse/6.9/graph/using/vertex-and-edge-ids.html}}) which are inherently non-contiguous.
\emph{Multi-level vectors} \cite{sortledton,gtx} mitigate the space issue by remapping arbitrary IDs into a compact range via hash tables, yet frequent reallocation and remapping operations during resizing can be a performance concern.
\emph{Tree-based structures} \cite{teseo,cpam,aspen} are the state-of-the-art approaches that offer logarithmic-time lookups and better space elasticity, which achieves a better trade-off between space and operational costs. 
In this paper, we show that we can achieve an even better space and performance trade-off, by redesigning radix trees which inherently avoid costly operations such as node split and merge in typical tree-based approaches, while our new design on trie structure selection renders a nice space optimization. 

\vspace{1mm}
\noindent\textbf{Challenges of edge index.} Existing dynamic graph systems generally adopt one of three edge storage paradigms: (1) Packed Memory Array (\emph{PMA-based}), (2) \emph{log-based}, or (3) \emph{tree-based}, each embodying distinct trade-offs among \textbf{update efficiency}, \textbf{efficient sequential scan}, and \textbf{space overhead}.
\emph{PMA-based designs}~\cite{pma,vcsr,10.1145/1292609.1292616} reserve gaps within arrays for edge insertions, and rebalance the gaps during updates. This offers excellent spatial locality and sequential scan performance but incurs non-trivial space and maintenance costs to manage reserved gaps within the array.
\emph{Log-based approaches}~\cite{gtx,livegraph,graphone} mainly store edge logs within an array, which delivers efficient sequential scan, low memory overhead and high insertion throughput via append-only updates. In this design, outdated entries must be located and invalidated through log traversal, incurring extra costs. 
\emph{Tree-based structures}~\cite{sortledton,gastcoco} organize adjacent edges of each vertex in a hierarchical layout chained with pointers, which offer a balanced trade-off between update efficiency and memory usage but reduce sequential scan performance compared to flat or contiguous storage.

\vspace{1mm}
\noindent\textbf{The problem.} \emph{Can we design a dynamic graph data structure that minimizes vertex space with high update/query efficiency, provides $O(1)$ edge operations, and preserves efficient sequential access?}

\vspace{1mm}
\noindent{\bf Our solution.} To address the problem, we propose RadixGraph, a highly performant for both updates and reads and space-efficient dynamic graph system that tackles the problem in the following ways:

\vspace{1mm}
\noindent{\bf SORT: a space-efficient and high-performance vertex index.} As an efficient data structure for indexing, the radix tree has emerged as a promising solution. The classical radix tree has a well-bounded query latency, coupled with its ability to avoid node splitting and merging when inserting vertices. Empirically, radix-tree-based solutions (Spruce \cite{spruce}, RadixGraph) show  higher vertex insertion throughput than other systems based on multi-level vectors or ARTs (Figure \ref{fig:dynamic-ops}(d)). Specifically, a recent work, Spruce~\cite{spruce}, employs the van Emde Boas tree (vEB-tree), which is a variant of the radix tree, to reduce the space cost by adjusting the fanouts of the trees. 
However, their adjustment is sub-optimal in the trade-offs between space cost and query performance, and how to optimally adjust the fanouts to achieve the best trade-off remains an open issue. 
To address this challenge, we firstly identify which possible radix tree structures can index a given ID universe, and formulate a constrained optimization program to select the best one that minimizes the space cost. Effectively solving the program gives us the space-optimized radix structure, dubbed as SORT. 

\vspace{1mm}
\noindent{\bf Snapshot-Log edge storage architecture: $O(1)$ operation time for insert, update, and delete with sequential access support.} 
Traditional approaches often face a trade-off between the efficiency of insertions, updates, deletions or sequential scans. To address this challenge, we introduce a snapshot-log edge storage architecture that separates the adjacency array into two distinct segments: a read-only snapshot and an updatable log. All updates and deletions are directly appended to the log segment as delta entries, which are then compacted into the snapshot segment once the log becomes full. By configuring the size of the new log segment after compaction to match that of the new snapshot, we ensure that the cost of this lazy compaction process remains bounded, achieving an $O(1)$ amortized cost per operation while preserving an optimal $O(d)$ get-neighbor query time and efficient sequential scan. To our best knowledge, this is the first edge structure that simultaneously supports $O(1)$ insertion, update, and deletion while maintaining sequential access. The $O(1)$ complexity is optimal and independent of vertex degree, allowing our design to remain both theoretically promising and practically simple.

\vspace{1mm}
\noindent{\bf Contributions.} 
We present RadixGraph, a novel in-memory data structure for dynamic graph storage that combines two innovative designs: a space-optimized vertex index (SORT) and a snapshot-log edge storage architecture. Empirically, we have used the widely adopted Graph Framework Evaluation (GFE) benchmark \cite{teseo} to compare RadixGraph against the state-of-the-art in-memory graph storage structures. Results demonstrate that RadixGraph achieves up to $16.27\times$ faster graph updates, and reduces memory consumption by an average of $40.1\%$ compared to the strongest baseline.

\section{Background}
\subsection{Radix tree and its variants}
\label{sec:radix-tree}
The radix tree (also known as trie) \cite{trie,knuth,Boehm2011EfficientII} is a highly efficient data structure for data management. Unlike other tree-like structures (e.g., B-tree \cite{Btree}, AVL tree \cite{AVL} and red-black tree \cite{RBtree}), the radix tree does not explicitly store keys in its internal nodes. Instead, it distributes a key into individual parts across multiple nodes in different layers to reduce redundancy. Advanced radix tree techniques mainly contain path-compressed radix trees \cite{PATRICIA,WORT,ART} and adaptive radix trees (ART) \cite{ART,DisART}.
\begin{figure}[t]
    \centering
    \includegraphics[width=1.\linewidth]{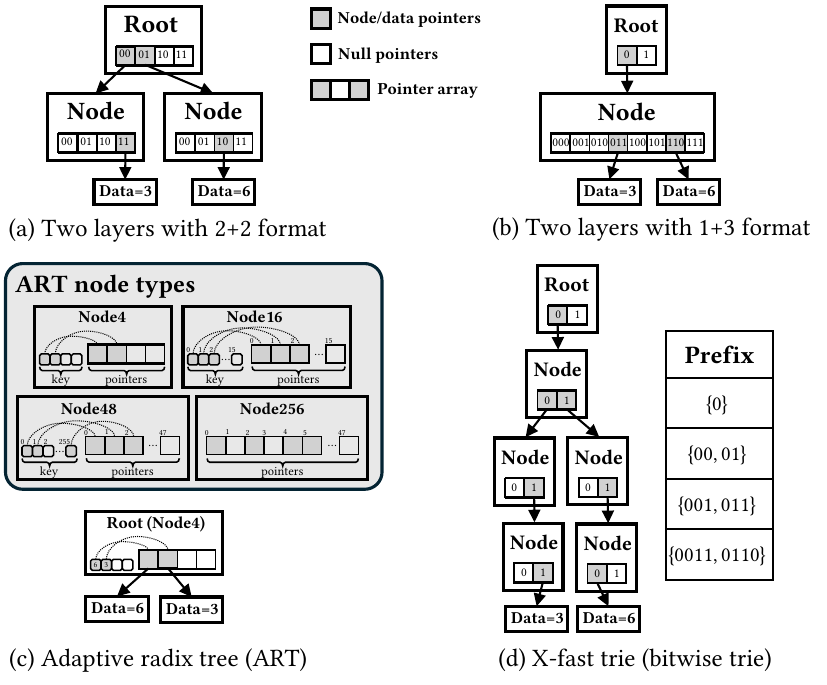}
    \caption{Different radix tree structures for storing 4-bit integers. Only integers 3 (i.e., \texttt{0011}) and 6 (i.e., \texttt{0110}) are inserted.}
    \label{fig:radixtree}
    \vspace{-1mm}
\end{figure}

\vspace{1mm}
\noindent\textbf{Canonical $l$-layer radix tree.} 
A canonical $l$-layer radix tree supports insertion and query operations in $O(l)$ time. 
Each layer partitions the key space by a fixed number of bits, and each node maintains a pointer array to their child nodes for further indexing. 
Specifically, if the $i$-th layer indexes $a_i$ bits, then each node in that layer maintains a pointer array of size $2^{a_i}$. 
Different ways of partitioning the key space lead to different radix tree structures. 
For example, Figures~\ref{fig:radixtree}(a) and~\ref{fig:radixtree}(b) illustrate two possible radix trees with $l=2$ for indexing 4-bit integers. 
To retrieve the integer $3$ (i.e., \texttt{0011}), the search in Figure~\ref{fig:radixtree}(a) follows branch \texttt{00} in the first layer and then \texttt{11} in the second layer. 
In contrast, the search in Figure~\ref{fig:radixtree}(b) follows branch \texttt{0} in the first layer and then \texttt{011} in the second layer. 
Given a key space $U = \{0, 1, \dots, u-1\}$, a radix tree may occupy up to $O(u)$ space. 
In practice, however, it remains space-efficient because only keys that actually occur are materialized. 
Nevertheless, the null pointers within the pointer arrays can still incur memory overhead. 
For instance, the empty branches \texttt{10} and \texttt{11} in the root node of Figure~\ref{fig:radixtree}(a) correspond to null pointers in the root node’s pointer array, even though they do not store any valid entries.

\vspace{1mm}
\noindent\textbf{Adaptive radix tree \cite{ART}.} 
To mitigate the space wasted by null pointers in pointer arrays, the adaptive radix tree (ART) was proposed. 
ART typically indexes 8 bits per layer. 
Instead of allocating a fixed pointer array of size $2^{8}$, ART defines four node types---\texttt{Node4}, \texttt{Node16}, \texttt{Node48}, and \texttt{Node256}---capable of storing up to 4, 16, 48, and 256 child pointers, respectively. 
The node type is dynamically chosen based on the number of non-empty child nodes. 
For example, Figure~\ref{fig:radixtree}(c) illustrates an ART instance for indexing 4-bit integers. 
Since ART indexes 8 bits per layer, only one layer is required, and because there are only two keys stored, a \texttt{Node4} is allocated. 
ART significantly reduces space consumption compared to canonical radix trees. 
However, it incurs additional lookup cost for \texttt{Node4} and \texttt{Node16} types, as locating a child requires scanning the whole pointer array. 
Furthermore, insertion throughput may degrade due to node resizing operations when a node’s capacity is exceeded.

\vspace{1mm}
\noindent\textbf{X/Y-fast trie~\cite{xfast}.} 
An X-fast trie is a special bitwise trie augmented with efficient successor query support. 
A bitwise trie corresponds to the case where $l = \lceil lg(u) \rceil$, where each layer indexes exactly one bit, as illustrated in Figure~\ref{fig:radixtree}(d). 
To enable fast successor queries, the X-fast trie maintains a hash table for each layer, storing all prefixes that appear at that level.  
However, in workloads that do not require successor queries, the maintenance of these hash tables introduces additional space and update overhead. 
A Y-fast trie~\cite{xfast} extends the X-fast trie by combining it with $O(n / lg(u))$ red–black trees to store $n$ elements. 
While this design improves space efficiency and supports $O(lglg(u))$ dynamic updates, the need to split and merge trees during updates can degrade practical performance and pose challenges to concurrency control.



\subsection{Existing solutions and limitations}
\label{sec:baselines}
Numerous studies have investigated dynamic graph storage. In this work, we focus on \textit{in-memory} dynamic graph systems, which typically consist of a vertex index and specially designed edge storage structures to support both vertex and edge updates.

\vspace{1mm}
\noindent
\textbf{Vertex index designs.} The vertex index aims to support fast vertex updates with minimal memory overhead. Most existing dynamic graph systems employ a large pre-allocated vertex array for indexing~\cite{stinger,graphone,livegraph,risgraph,gastcoco,llama}, typically of size $2^{32}$ that corresponds to the common range of vertex identifiers. However, this approach results in substantial memory waste, as vertex identifiers in dynamic graphs are often non-contiguous, leading to fragmentation and inefficient space utilization. To address this issue, recent systems adopt adaptive radix trees (ART)~\cite{teseo}, van Emde Boas trees (vEB-trees)~\cite{spruce}, or multi-level vector indices (often integrated with an external hash map)~\cite{sortledton,gtx}. Although these adaptive structures provide more flexible indexing, they may incur additional time overhead from frequent structural modifications during insertions. For example, multi-level vector indices trigger hash map resizing and rehashing once the load factor is exceeded, while ART must resize and migrate node arrays when a node becomes full. In contrast, the vEB-tree achieves high throughput with its fixed structure, but this rigidity gives further room to optimize space efficiency and adaptability to different ID spaces.

\vspace{1mm}
\noindent\textbf{Adjacency list-based edge structures.} Dynamic graph systems based on adjacency lists primarily address two key challenges: (1) improving the efficiency of edge insertions and deletions, and (2) linking edge blocks to enhance read performance. Stinger~\cite{stinger} organizes edges into blocks to improve read efficiency. GraphOne~\cite{graphone} adopts a hybrid approach by combining an adjacency list with an edge log list, where new edges are initially logged and later batch-materialized into the adjacency list. Similarly, LiveGraph~\cite{livegraph} and GTX~\cite{gtx} introduce the Transactional Edge Log (TEL) to manage edge modifications. However, they require traversing the entire adjacency list to locate and remove edges, leading to $O(d)$ complexity per update or deletion, where $d$ is the degree of the vertex. To enhance efficiency, RisGraph~\cite{risgraph} employs an indexed adjacency list combined with sparse arrays, significantly improving read performance while reducing traversal overhead. Sortledton~\cite{sortledton} and GastCoCo~\cite{gastcoco} replace traditional edge lists with unrolled skip lists and B+ trees, respectively, achieving efficient edge insertions, updates, and deletions in $O(lg(d))$ time complexity. Similarly, Aspen \cite{aspen} and CPAM \cite{cpam} introduce C-tree, an improved version of the B+-tree that also supports $O(lg(d))$ time complexity, which enhances cache locality and is optimized for parallel batch-updates. Despite their efficient updates, these approaches introduce additional index structures for managing edge lists, which increase overall space consumption. Spruce~\cite{spruce} buffers newly inserted edges in a fixed-size array, which is merged into the sorted array once full. While this approach achieves high space efficiency, the merging process introduces an amortized $O(d)$ cost per edge insertion.

\vspace{1mm}
\noindent\textbf{CSR-based edge structures.} Dynamic graph systems based on CSR primarily aim to support efficient edge insertions and deletions while preserving cache-friendly access patterns. CSR++~\cite{firmli_et_al:LIPIcs.OPODIS.2020.17} introduces a segmented CSR representation that divides the edge array into multiple segments, which reduces the need for full reconstruction. However, it still includes local shifting costs within segments during updates. LLAMA~\cite{llama} takes a versioned approach by maintaining multiple graph snapshots, enabling efficient temporal queries but at the cost of increased memory usage. Another line of work replaces CSR’s contiguous edge array with a Packed Memory Array (PMA)~\cite{pma,vcsr,10.1145/1292609.1292616}, which maintains edges in a partially sorted, dynamically resizable array. For instance, Teseo~\cite{teseo} stores PMAs in the leaf nodes of its vertex index, while Terrace~\cite{terrace} combines sorted arrays, PMAs, and B+ trees to handle vertices of varying degrees. PPCSR \cite{ppcsr} further enhances the PMA to support concurrent operations. Compared with previous approaches, PMA provides efficient updates with complexity guarantees while preserving sequential access. The main challenge for PMA is the practical trade-off between space consumption and update efficiency. Under highly dynamic workloads, PMA requires rebalancing and can fragment memory or temporarily reduce update throughput. Although the theoretical space of most edge structures is $O(m)$, the practical trade-off may result in difference in actual space consumption.
\section{Design of RadixGraph}
Our graph system, RadixGraph, identifies a space-optimized radix tree (SORT) as the vertex index and stores the vertex information in a vertex table. RadixGraph also facilitates a novel snapshot-log edge structure supporting $O(1)$ time for edge insertions, updates and deletions. Figure \ref{fig:structure} shows a high-level overview of RadixGraph. In this example, there are 5 vertices with IDs 3, 49, 2, 52 and 1 inserted sequentially, and their information is stored in the vertex table. Since the vertex IDs are 6-bit integers, the SORT has 3 layers that indexes 3 bits, 2 bits and 1 bit, respectively.  The SORT indexes these vertex IDs and offers their locations (i.e., offsets) in the vertex table. For example, the binary format of ID 52 is \texttt{110100}, and its offset (96) can be thereby retrieved by traversing these bits in SORT. Each vertex is associated with an edge array which is equally partitioned into a snapshot segment and a log segment, where the snapshot segment stores the edges and the log segment stores the recent edge updates of that vertex. For example, in the edge array of vertex 49, the log block $(null, (64), 126)$ represents deleting the edge from vertex 49 to 2 at timestamp 126.

In the following, we will show the vertex storage structure in Section \ref{sec:vertex-index} and how to optimize SORT in Section \ref{sec:optimize-SORT}. Then we introduce the edge storage in Sections \ref{sec:edge-structure}, respectively. Moreover, we discuss the complexities and concurrency control in Section \ref{sec:discuss}.
\begin{figure*}[ht]
    \centering
    \includegraphics[width=1.\linewidth]{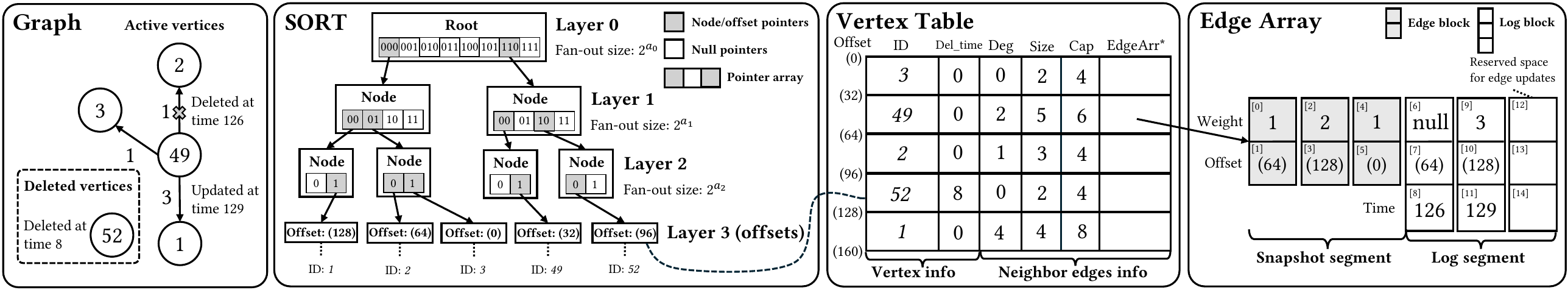}
    \caption{The high-level overview of RadixGraph. Parenthesis ``(x)'' represents an offset corresponding to the location in the vertex table. The fan-out size $2^{a_i}$ is the pointer array size of a node at layer $i$ and is pre-determined by our optimizer. In this example, $a_0=3,a_1=2,a_2=1$.}
    \label{fig:structure}
    \vspace{-2mm}
\end{figure*}
\subsection{SORT: a Space-Optimized Radix Tree}
\label{sec:vertex-index}
\noindent\textbf{Motivation.} RadixGraph employs a radix-tree-like structure to map vertices to their locations in the vertex table. As discussed in Section~\ref{sec:radix-tree}, there are multiple possible structures for a radix tree. For canonical $l$-layer radix trees, a common practice is to assign the same fan-out size to different layers, and we denote it as uniform-tree. Given $n$ integers in an integer universe $U=\left\{0,1,\cdots,u-1\right\}$ and the number of layers $l$, the fan-out size of a uniform-tree is $2^{\left\lceil lg(u)/l\right\rceil}$. Moreover, the Van Emde Boas tree (vEB-tree) \cite{vEB} recursively partitions the universe into $O(\sqrt{u})$ subtrees. For instance, when $u=2^{64}$, the top-level fan-out is $2^{32}$, followed by a second layer with fan-out $2^{16}$, leading to a depth of $O(lglg(u))$. However, both of them apply a rigid fan-out size across layers, which often fails to strike an optimality balance between space usage and efficiency. To address this, RadixGraph introduces a novel \textbf{S}pace-\textbf{O}ptimized \textbf{R}adix \textbf{T}ree (SORT) that determines the optimal fan-out at each layer according to graph data, which minimizes the average space costs given the number of layers $l$ of the radix tree. In Table \ref{tab:radix-tree-tradeoff}, we set the same $l$ for different radix tree structures to align the query efficiency, which is $O(l)$, and compare the performances of them. SORT exhibits the least memory usage since it optimizes the space cost, and this also benefits the insertion efficiency. For ART, although it introduces flexible node types to index its children and exhibits competitive space consumption compared with SORT, its efficiency is slower due to the need to transform between these node types. Detailed comparisons and analysis are provided in Sections \ref{sec:case-study} and \ref{sec:ablation}. X/Y-fast trie can be viewed as an extended bitwise trie, which fits the canonical $l$-layer radix tree case when $l=O(lg(u))$. Therefore, the SORT optimization scheme can also be integrated for these variants. However, they are mainly optimized for successor queries and require maintaining extra components to support them, which are redundant for our graph storage task since we do not require successor queries.
\begin{table}[t]
    \centering
    \fontsize{7pt}{1em}\selectfont
    \begin{tabular}{c|c|c|c|c|c|c}
    \hline
     \multirow{2}{*}{$n$} & \multicolumn{2}{c|}{Uniform-tree} & \multicolumn{2}{c|}{vEB-tree} & \multicolumn{2}{c}{SORT}\\
     \cline{2-7}

     & Memory & Insertion & Memory & Insertion & Memory & Insertion \\
    \hline\hline
    $10^3$ & 4.38MB & 6.8ms & 2.62MB & 4.7ms & \textbf{\underline{0.85MB}} & \textbf{\underline{2.8ms}} \\
    \hline
    $10^4$ & 38.11MB & 29ms & 20.41MB & 23ms & \textbf{\underline{5.23MB}} & \textbf{\underline{16ms}} \\
    \hline
    $10^5$ & 295.60MB & 172ms & 119.12MB & 118ms & \textbf{\underline{32.09MB}} & \textbf{\underline{86ms}} \\
    \hline
    \end{tabular}
    \caption{Performances of different $O(lglg(u))$-layer radix tree structures for inserting $n$ random IDs in $[0,2^{32}-1]$.}
    \label{tab:radix-tree-tradeoff}
    \vspace{-1mm}
\end{table}

\vspace{1mm}
\noindent{\bf Structure of SORT.} Figure \ref{fig:structure} shows the high-level structure of a SORT. Specifically, if the depth of the tree is set as $l$ ($l\geq 1$), then SORT begins with a root node at layer 0, and each node in layer $i$ has an array of $2^{a_i}$ pointers, where $a_i$ would be pre-determined by the optimizer, which will be introduced shortly in Section \ref{sec:optimize-SORT}. Each pointer is either a null pointer or points to a child node at layer $(i+1)$. Specifically, the leaf 
nodes at layer $l$ store the byte offsets of their corresponding vertices, which correspond to their locations in the vertex table.

\begin{algorithm}[ht]
    \caption{Insert-Vertex($N,i,v,a[]$)}
    \label{alg:insert}
    
    \KwIn{the current tree node $N$ and its layer $i$ and the vertex identifier $v$ in binary format; $a_i$ is pre-determined by the optimizer.}
    
    $x\gets$ the first $a_i$ bits of $v$\;
    \If{$N.children[x]$ is a null pointer} {
        \If{$i<l-1$}{
            Create a child node $c$, initialize its pointer array of size $2^{a_{i+1}}$ and store the pointer of $c$ in $N.children[x]$\;
        }
        \Else{
            Store the vertex in a new entry of the vertex table\;
            Store the offset of the vertex in $N.children[x]$\;
            \textbf{return}\;
        }
    }
    Delete the first $a_i$ bits of $v$\;
    Insert-Vertex($N.children[x],i+1,v,a$)\;
\end{algorithm}

Algorithm \ref{alg:insert} details the process for inserting a vertex into SORT. The idea is simple yet effective: the vertex identifier is divided into individual segments in its binary format, and each segment is processed in a corresponding layer of the tree. At each layer, the value of the segment determines the specific child node to which the vertex is assigned. If the corresponding child node has not been created, we create a child node and store its pointer. The process continues recursively until the vertex is created and its offset is stored at the $l$-layer in the tree.

The process of retrieving a vertex from SORT follows a similar approach to Algorithm \ref{alg:insert}, which costs $O(l)$ time. By traversing the tree based on the segments of the vertex identifier, the algorithm either returns the desired vertex offset or returns a null value (e.g., -1) if the vertex has not been stored.

\begin{figure}
    \centering
    \includegraphics[width=1.\linewidth]{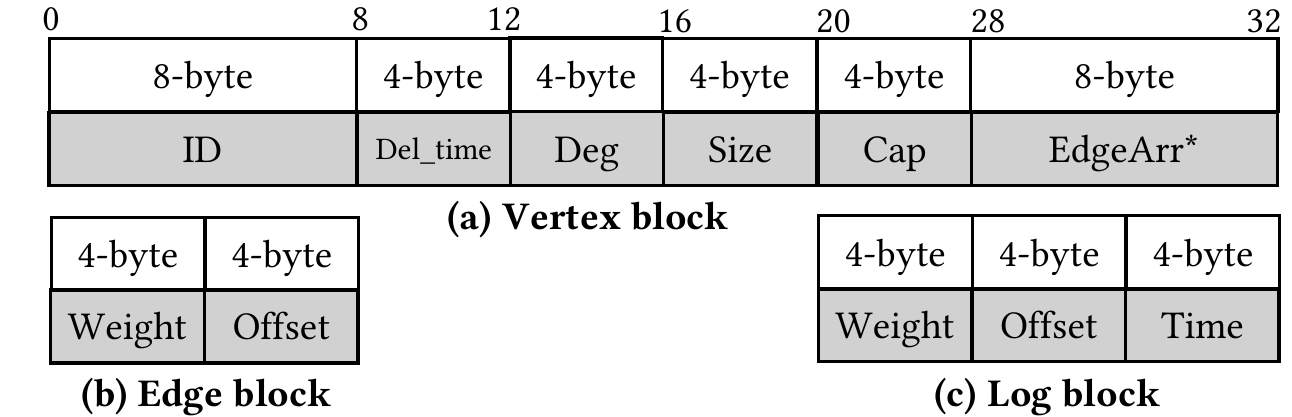}
    \caption{The data layout of vertex, edge and log blocks.}
    \label{fig:layout}
    \vspace{-1mm}
\end{figure}
\vspace{1mm}
\noindent\textbf{Vertex table.} RadixGraph maintains an expandable vertex table with segmented storage, implemented using Intel's TBB concurrent vector \cite{intel_tbb}. When a segment reaches capacity, a new segment is allocated with twice the size of the previous one. The previous segments are kept such that related updates on vertices within them can process without being blocked. This segmented design ensures that once a vertex is inserted, its physical location remains fixed, avoiding data movement and preventing read-write conflicts. Figure~\ref{fig:layout} illustrates the layout of a vertex block, which consists of the vertex ID (i.e., \textit{ID}), the deletion timestamp (i.e., \textit{Del\_time}), the degree of the vertex (i.e., \textit{Deg}), the number of occupied blocks in the edge array (i.e., \textit{Size}), the total number of blocks in the edge array (i.e., \textit{Cap}), and a pointer to the edge array (i.e., \textit{EdgeArr*}). By default, \textit{Del\_time} is initialized to 0, indicating the vertex is active. When a vertex is deleted, RadixGraph sets \textit{Del\_time} to the current timestamp $t$, making the vertex invisible to transactions with timestamps greater than $t$. 

For garbage collection, the offsets of deleted vertices are stored in a queue. Figure \ref{fig:structure} shows an example where vertex 52 is deleted. The queue will store its offset (96), and when there is a vertex insertion, the offset is retrieved from the queue and the new vertex will reuse this offset and its corresponding slot in the vertex table. More specifically, when a new vertex is inserted, the system first checks the queue for reusable slots via atomic compare-and-swap (CAS) and only expands the vertex table when there are no available slots. Deleted vertices are only purged from the queue when all transactions before ``Del\_time'' are finished for MVCC consistency.
\subsection{Optimizing SORT configuration}
\label{sec:optimize-SORT}
\noindent\textbf{Problem definition.} Let the vertex identifier \( X \) be a discrete random variable uniformly distributed over the integer domain \(\{0, 1, \ldots, 2^{x}-1\}\), denoted as \( X \sim \mathcal{U}_d(0, 2^{x}-1) \).
Given \( n \) distinct identifiers sampled from this domain and a fixed number \( l \), we aim to find an optimal radix tree configuration within the family of canonical \( l \)-layer radix trees.
Formally, our objective is to determine the structure \( T^{*} \in \mathcal{T}_{l} \) that minimizes the average space consumption:
$$
T^{*} = \arg\min_{T \in \mathcal{T}_{l}} \; \mathbb{E}_{X \sim \mathcal{U}_d(0, 2^{x}-1)}[\text{Space}(T, X, n)],
$$
where \(\mathcal{T}_{l}\) denotes the set of canonical \(l\)-layer radix trees capable of indexing \(n\) distinct \(x\)-bit integers. Although our analysis assumes a uniform integer domain, the proposed method remains effective under skewed workloads. 
Empirical evaluation on such workloads is provided in Section~\ref{sec:case-study}.

\vspace{1mm}
\noindent\textbf{Overview.} We now present an overview for determining the optimal structural configuration of SORT. 
The objective is to minimize the average space cost of the tree under distinct and uniformly distributed integer identifiers, formulated as:
$$
\min_{a_i \in \mathbb{N}} \quad \sum_{i=0}^{l-1} 2^{a_i} \cdot N(i) \cdot p(i),
\quad \text{s.t.} \quad 2^{a_0 + a_1 + \cdots + a_{l-1}} \ge 2^{x}
$$
Here, \(a_i\) is the logarithmic fan-out size of layer \(i\) for the SORT structure. Given $a_i$ for all $0\leq i\leq l-1$, \(N(i)\) is defined as the maximum number of nodes that can be instantiated in layer \(i\), and \(p(i)\) denotes the probability that a node at layer \(i\) is instantiated. The optimization requires input $l,x$ and $n$, where $l$ is the number of layers, \(x\) is the bit-length of each vertex identifier and \(n \ge 1\) is the total number of vertices.

The objective function expresses the total space consumption of SORT as a layer-wise summation, where \(N(i) \cdot p(i)\) represents the expected number of instantiated nodes in layer \(i\), and \(2^{a_i}\) corresponds to the pointer array size of each node. 
Thus, the objective is the expectation of \(\text{Space}(T, X, n)\), representing the average space usage of a canonical \(l\)-layer radix tree under uniform identifiers.
The constraint defines the feasible space \(\mathcal{T}_l\), ensuring that SORT spans the entire identifier domain \([0, 2^{x}-1]\).  
In the following, we derive analytical formulations for \(N(i)\) and \(p(i)\), which enable a closed-form optimization of the above objective.
\begin{figure}[t]
    \centering
    \includegraphics[width=1.\linewidth]{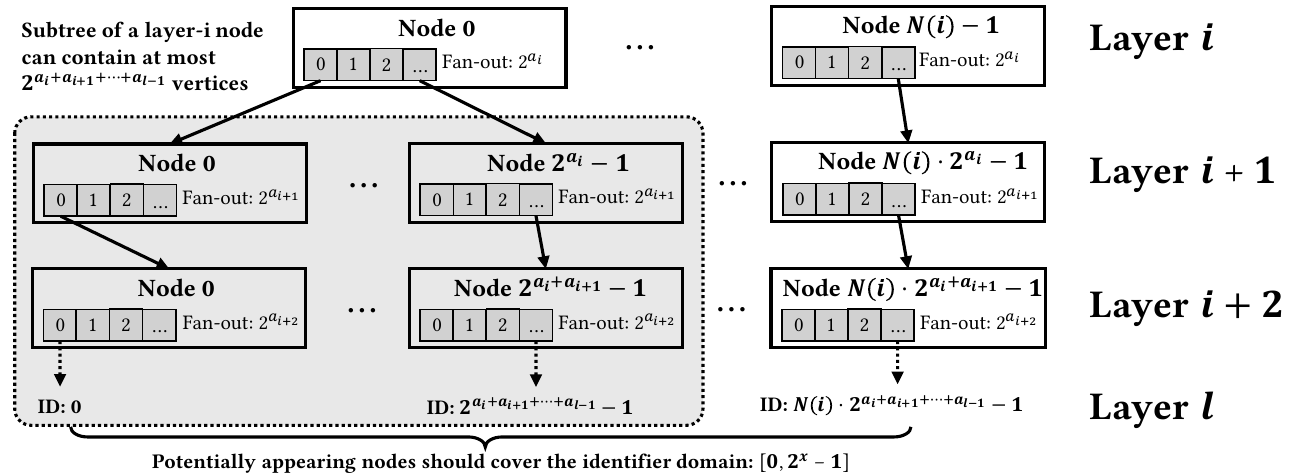}
    \caption{The representation intervals of $i$-layer nodes.}
    \label{fig:node-interval}
    \vspace{-1mm}
\end{figure}

\vspace{1mm}
\noindent{\bf Objective function formulation.} For layer $i$, as shown in Figure \ref{fig:structure}, each internal node contributes to a pointer array of size $2^{a_i}$. Specifically, for layer 0, since $n\geq 1$, the root node must be created and contain an array of $2^{a_0}$ child pointers, and $N(0)=1,p(0)=1$. For a node at layer $i$ ($i>0$), its subtree can store at most $2^{a_i+a_{i+1}+\cdots+a_{l-1}}$ graph vertices, whose IDs span an interval $[L,R]$, and all intervals of nodes at layer $i$ are pairwise disjoint. For example, as shown in Figure \ref{fig:node-interval}, where layer 0 to $l-1$ stores internal tree nodes and layer $l$ stores graph vertices, the internal node 0 in layer $i$ covers graph vertices within an ID interval $[0,2^{a_i+a_{i+1}+\cdots+a_{l-1}}-1]$, which are stored in layer $l$. For simplicity, we denote:
$$S_i=R-L+1=2^{a_i+a_{i+1}+\cdots+a_{l-1}}$$
as the length of the node interval for layer $i$. Therefore, the root node (or equivalently, the whole SORT) can accommodate at most $S_0=2^{a_0+a_1+\cdots+a_{l-1}}$ graph vertices, which should cover the identifier domain $[0,2^x-1]$ and thus we derive the constraints of the optimization.

We observe the optimal configuration must satisfy $S_i\leq 2^x$ for $0\leq i<l$. Otherwise, the covered ID interval of the leftmost node at layer $i$ can already represent all the vertices within $[0,2^x-1]$. Then $N(i)=1$ and $p(i)=1$, and reducing $a_i$ to $a_i-1$ can result in a smaller objective value. Given the form of $S_i$ and the fact that $S_i\leq 2^x$, this indicates that $S_i$ always divides $2^x$ and a layer $i$ tree node can thereby be classified as two cases: (1) its interval lies entirely within the domain $[0,2^x-1]$; (2) its interval lies completely outside this domain.

For case (2), the node would never be created since its interval exceeds the range of vertex identifiers, and thus does not contribute to the space cost. Therefore, the maximum number of tree nodes $N(i)$ that can appear in layer $i$ corresponds to the nodes of case (1):
$$
N(i)=\frac{2^x}{S_i}=2^{x-(a_{i}+a_{i+1}+\cdots+a_{l-1})}
$$

For case (1), we consider the complement event that the node is {\bf not} created, which means that all $n$ vertices are distributed in other $(2^x-S_i)$ possible IDs out of its corresponding interval. Since all vertex identifiers are distinct integers, the probability of the complement event is $\left(\begin{matrix}
    2^x-S_i\\
    n
\end{matrix}\right){\bigg/}\left(\begin{matrix}
    2^x\\
    n
\end{matrix}\right)$ by hypergeometric distributions \cite{feller1}, and we immediately have the probability of this node being created as:
$$
    p(i)=1-\left(\begin{matrix}
    2^x-S_i\\
    n
\end{matrix}\right){\bigg/}\left(\begin{matrix}
    2^x\\
    n
\end{matrix}\right)=1-\frac{(2^x-2^{a_i+\cdots+a_{l-1}})!(2^x-n)!}{(2^x)!(2^x-2^{a_i+\cdots+a_{l-1}}-n)!}
$$
where $\left(\begin{matrix}
    a\\
    b
\end{matrix}\right)=\frac{a!}{b!(a-b)!}$ is the combination number representing the number of different ways to take $b$ items from $a$ distinct items. Note that $\left(\begin{matrix}
    a\\
    b
\end{matrix}\right)=0$ when $a<b$, meaning $p(i)=1$ when $2^x-S_i<n$.

Putting everything together, now we are ready to formulate the complete optimization problem:
\[
\min_{a_i \in \mathbb{N}} 2^{a_0}+\sum_{i=1}^{l-1} 2^{x-(a_{i+1}+\cdots+a_{l-1})}\left(1-\frac{(2^x-2^{a_i+\cdots+a_{l-1}})!(2^x-n)!}{(2^x)!(2^x-2^{a_i+\cdots+a_{l-1}}-n)!}\right)
\]
$$
\text{s.t.} \quad a_0+a_1+\cdots+a_{l-1}\geq x
$$

\vspace{1mm}
\noindent{\bf Solving the optimization problem.} Directly solving the optimization problem via integer programming incurs exponential time complexity. Fortunately, by exploiting structural properties of the objective function, the problem can be simplified and solved in polynomial time.

To simplify the optimization problem, we define $s_i=\sum_{j=0}^ia_j$. Then $s_j\in\mathbb N,s_{l-1}\geq x$ and $s_j\leq s_{j+1}$ for all $0\leq j<l-1$. We can rewrite the objective function as:
$$
    f(s_0,\cdots,s_{l-1})=2^{s_0}+\sum_{i=1}^{l-1}\frac{2^x}{2^{s_{l-1}-s_i}}\left(1-\frac{(2^x-2^{s_{l-1}-s_{i-1}})!(2^x-n)!}{(2^x)!(2^x-2^{s_{l-1}-s_{i-1}}-n)!}\right)
$$

We now present its optimality condition in Lemma \ref{lemma:optimality}, which shows the optimal tree structure should fit into the interval $[0,2^{x}-1]$ exactly. That is, its representation range should be $[0,2^x-1]$.
\begin{lemma}
    \label{lemma:optimality} (Proof in Appendix A)
    There exists an optimal solution $(s_0^*,s_1^*,\cdots,s_{l-1}^*)$, such that $s_{l-1}^*=x$.
\end{lemma}

Then our problem becomes how to select $s_0,s_1,\cdots,s_{l-2}$ to optimize the internal structure of the tree, and the objective function $f$ can be simplified by substituting $s_{l-1}$ with $x$:
$$
f(s_0,\cdots,x)=2^{s_0}+\sum_{i=1}^{l-1}2^{s_i}\left(1-\frac{(2^x-2^{x-s_{i-1}})!(2^x-n)!}{(2^x)!(2^x-2^{x-s_{i-1}}-n)!}\right)
$$

We observe that for the $i$-th summation term in $f$, its value only depends on $s_i$ and $s_{i-1}$, which inspires us to solve the problem via \textbf{dynamic programming} that recursively solves the optimal value of the current state $(s_i)$ from optimal previous states $(s_{i-1})$'s.

Specifically, we define an auxiliary function $g$, such that $g(i,j)$ represents the minimum average space cost for the first $(i+1)$ layers when $s_i=j$. In fact, $g(i,j)$ is the optimal value of considering $2^{s_0}$ and the summation of first $i$ terms in $f$ when $s_i=j$. Therefore, we could solve $g(i,j)$ via dynamic programming, and the transition is given by:
\begin{equation}
    \label{eq:transition}
    g(i,j)=\min_{0\leq k\leq j}g(i-1,k)+2^{j}\left(1-\frac{(2^x-n)!(2^x-2^{x-k})!}{(2^x)!(2^x-2^{x-k}-n)!}\right)
\end{equation}
with the initial states $g(0,j)=2^{j}$ for all $0 \le j \le x$. The main idea of this dynamic programming, is to partition the tree by layers. For each layer $i$, we enumerate all possible values $k$ for $s_{i-1}$ and values $j$ for $s_i$. Then we can derive the optimal inductions from layer $i-1$ to layer $i$ for all $s_i=j$. The initial state $g(0,s_0)$ corresponds to $2^{s_0}$ and thereby equals $2^j$ when $s_0=j$. Inductively solving $g(i,j)$ from $i=0$ to $i=l-1$ gives us the globally optimal solution.

\begin{theorem}
    (Proof in Appendix A) The optimal value of $f$ is equal to $g(l-1,x)$.
    \label{theorem:optimality}
\end{theorem}

\vspace{1mm}
\noindent\textbf{Tuning the depth of SORT.} Generally, the depth $l$ of the tree should not exceed $x$. Otherwise, there exists $a_i=0$ for some $i$ and these layers can be safely pruned to reduce space costs as discussed above. For dense universe (i.e., $n\approx 2^x$), the optimal $l$ tends to be smaller since most possible identifiers are present, and we can afford a large fan-out (i.e., large $a_i$) without wasting much space. For sparse universe (i.e., $n\ll 2^x$),  the optimal $l$ tends to be larger and $a_i$ should be smaller to reduce excessive empty pointers. Specifically, when $n=2^x$, the optimal setting for minimizing space is to set $l=1$ and $a_0=x$, where the total space is $O(n)$. Since the lookup efficiency of the vertex index depends on $l$, in practice, we can set $l$ based on the required lookup efficiency to compute an optimal setting of $a_i$, and then prune all layers with $a_i=0$ to derive an appropriate setting for the tree. For example, in our experiments, we set $l=O(lglg(u))$, where $[0,u-1]$ is the vertex ID universe to ensure a fair comparison with Spruce, which also utilizes a radix-tree-like structure (vEB-tree) with query complexity $O(lglg(u))$.

\vspace{1mm}
\noindent\textbf{Implementation and adaptation.} Given input $n$ (the number of integers) and $x$ (the bit-length of vertex IDs), we implement the iterative process in Equation \ref{eq:transition} by for-loop structures, where the outermost loop enumerates $i$ from $0$ to $l-1$, the middle loop enumerates $j$ from $0$ to $x$, and the innermost loop enumerates $k$ from $0$ to $j$. This implementation is efficient and takes less than 1 second (5{\textperthousand} of the graph construction time) to compute the optimal setting on our largest experimental dataset \textit{twitter-2010}. As $n$ evolves, the optimal configuration of SORT may gradually shift. However, in real-world graphs, vertex insertions are relatively infrequent compared to edge updates. For instance, Twitter’s daily active users increased from 115 million to 237.8 million between 2017 and 2023\footnote{\url{https://famewall.io/statistics/twitter-stats}}, less than a double growth over six years. Empirically, we observe that SORT’s optimal configuration remains largely stable with respect to $n$; the adjustment intervals grow roughly exponentially with powers of 2 in $n$ (see Section~\ref{sec:case-study}). Moreover, the memory overhead remains small even if the configuration is not updated immediately.

To minimize disruption, we recompute the optimal configuration asynchronously and trigger an update only when $n$ changes exponentially. When an update is required, we employ a lazy transformation strategy: subtrees with unchanged fanouts are preserved by reusing their existing pointers, while only the affected upper levels are reconstructed. As shown in Section~\ref{sec:case-study}, even when SORT is aggressively adapted whenever its configuration changes, each transformation completes within about 1 second, which clearly shows the transformation costs are negligible relative to overall system performance.

\vspace{1mm}
\noindent\textbf{Supplementary information.} Due to the space limit, we refer readers to the appendices for detailed illustrations, proofs and complexity analysis.

\subsection{Snapshot-log edge storage}
\label{sec:edge-structure}
RadixGraph adopts an enhanced log-based structure for edge storage, since it exhibits high edge insertion throughput. Other log-based graph systems such as LiveGraph and GTX \cite{livegraph,gtx} also achieve constant-time edge insertions by appending edges to a per-vertex log array, and they support multi-version concurrency control (MVCC) by maintaining historical edge versions. However, this comes at the cost of low update and deletion efficiency, as each modification requires traversing the whole log array to locate and invalidate the previous version.

RadixGraph maintains an edge array per vertex that is evenly partitioned into a snapshot segment and a log segment. For instance, in Figure \ref{fig:structure}, vertex 49 maintains an edge array of 6 blocks, with 3 blocks in the snapshot segment and 3 in the log segment. The key benefit of this layout is that it enables an out-of-place update strategy which avoids costly traversal during updates: edge updates are first appended to the log segment and are later compacted in batches. When the log segment becomes full (i.e., when the ``Size'' field reaches the ``Cap'' field), a compaction is triggered that combines the snapshot and log segments to produce a new snapshot segment containing only the latest versions of edges. The size of this snapshot segment equals the vertex degree $d$ (recorded in the ``Deg'' field of the vertex block). Since the entire edge array has capacity $2d$, the number of outdated entries is always bounded by $O(d)$. Consequently, all edge operations run in amortized $O(1)$ time: appending an update to the log segment costs $O(1)$, and although each compaction takes $O(d)$ time, this cost is amortized over the $d$ updates processed by the compaction. As a result, the amortized cost of compaction is $O(1)$ per update.

\vspace{1mm}
\noindent\textbf{Edge array.} Each vertex is associated with an edge array composed of two segments: a snapshot segment storing edge blocks and a log segment storing log blocks. Figure~\ref{fig:layout} illustrates the layout of these edge and log blocks. Each edge block stores its metadata (e.g., \textit{Weight}) along with an \textit{Offset} field. The \textit{Offset} field stores the destination vertex's position in the vertex table to track its information, and source vertex is not stored in the edge and log blocks since its information is already presented in the corresponding vertex block. For log blocks, an additional \textit{Time} field is included to support edge versioning. Specifically, if the log entry represents an edge insertion or update, the \textit{Weight} field holds the updated weight. If the log entry corresponds to a deletion, the \textit{Weight} field is set to NULL (e.g., 0). For each operation, the source and destination vertices are first located or inserted into SORT and the vertex table. Then, a log entry is appended to the log segment according to the operation type. Therefore, each directed edge consumes only 1 block; for undirected edges, it will consume 2 blocks since both directions of the edge are inserted.

\begin{itemize}[leftmargin=*] \item Insertion: Adds an edge to the graph, appending a log with the edge's properties and the byte offset of the destination vertex; \item Update: Modifies the properties of an existing edge, appending a log with the edge's updated properties and the byte offset of the destination vertex; \item Deletion: Removes an edge, appending a log with a NULL property and the byte offset of the destination vertex. \end{itemize}

Figure \ref{fig:compaction} illustrates an example of the edge update process, where the insertion of edge from vertex 49 to vertex 2 is firstly appended to the log segment and then triggers a compaction since the edge array is full. More specifically, when the log segment is not full, the update process is lock-free and costs $O(1)$ time: each edge log obtains a unique slot via an atomic \textit{fetch\_add} on the \textit{Size} field. When the edge array becomes full (i.e., ``Size''==``Cap''), a lock is acquired to perform log compaction. Specifically, if the current degree of the vertex is $d$, the capacity of the new array after compaction is $2d$, reserving space for $d$ additional log entries.

\begin{figure}[t]
    \includegraphics[width=0.98\linewidth]{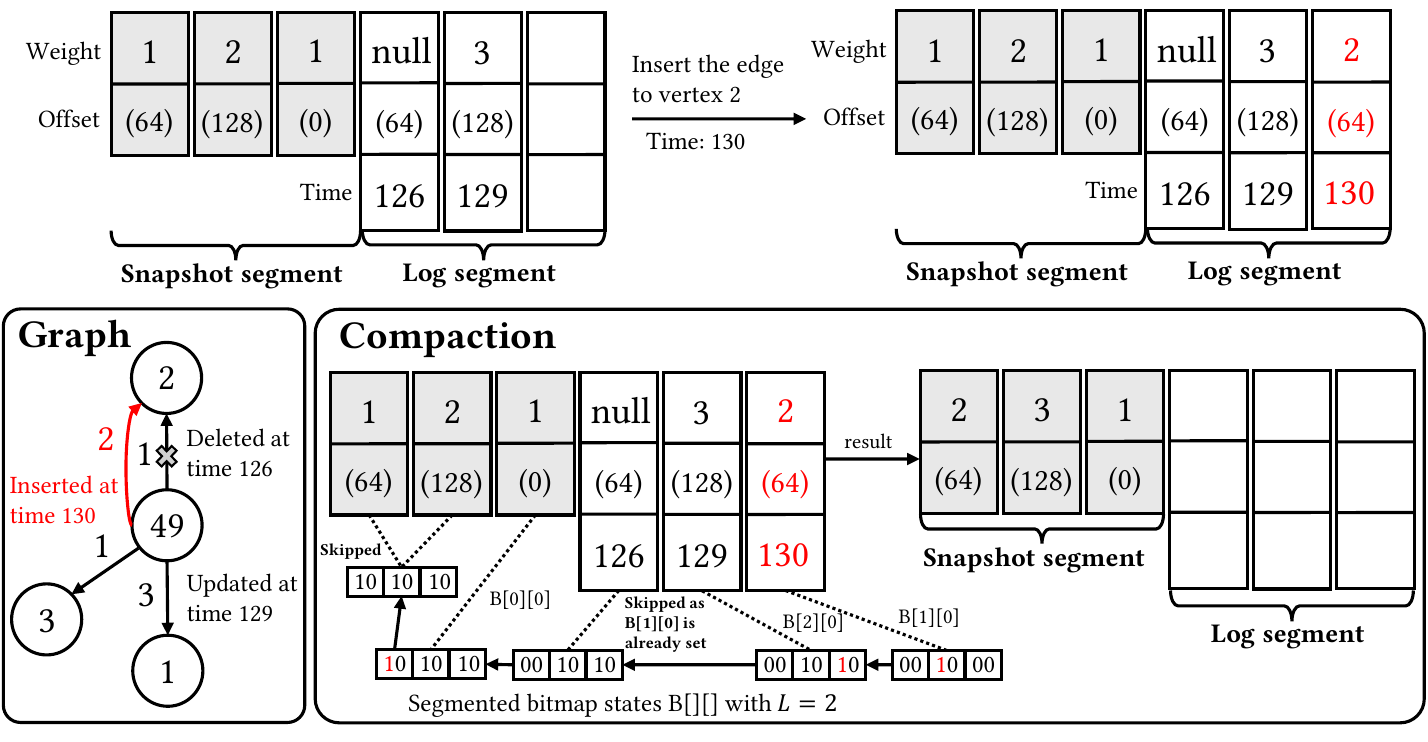}
    \caption{Example edge update and compaction processes. Offset (x) corresponds to the bit $B[(x/32)/L)][(x/32)\text{ mod } L]$.}
    \label{fig:compaction}
    \vspace{-1mm}
\end{figure}

\vspace{1mm}
\noindent\textbf{Log compaction.} As edge logs are appended independently and may create multiple entries for the same destination, the compaction procedure must identify and retain only the latest valid edge among these duplicates. To address this, we introduce \textit{duplicate checker} for each thread to efficiently identify the latest version of an edge. Algorithm \ref{alg:log-compaction} shows the process. During log compaction, we perform a reverse iteration on the array, such that edge logs are traversed from most recent ones to least recent ones. For each edge log, if the duplicate checker finds that its destination vertex has not been visited, this means the edge is the latest version in the current graph, and the edge is thereby added to $A$; otherwise, it is skipped. After processing each edge log, the duplicate checker sets the destination vertex as visited. This ensures outdated edge logs to be skipped and only latest edges are retrieved from the array. After the whole edge array is processed, we reset all vertices as unvisited in the duplicate checker.

\begin{algorithm}[ht]
    \caption{LogCompaction($O_u$)}
    \label{alg:log-compaction}
    
    \KwIn{the target vertex offset $O_u$.}
    $C\gets$ the duplicate checker of current thread\;
    $A\gets$ an edge array of size $2\times V[O_u].Deg$\;
    $E_u\gets V[O_u].EdgeArr,s\gets V[O_u].Size,cnt\gets 0$\;
    \For{$i=s-1,s-2,\cdots,0$} {
        $O_v\gets E_u[i].Offset$\;
        \If{$C$ finds $v$ not visited and $E_u[i].Weight\not=null$} {
            $A[cnt\mathbin{++}]\gets\left\{E_u[i].Weight,O_u\right\}$\;
        }
        Mark $v$ as visited in $C$\;
    }
    \For{$i=0,1,\cdots,s-1$} {
        $O_v\gets E_u[i].Offset$\;
        Mark $v$ as unvisited in $C$\;
    }
    $V[O_u].EdgeArr\gets A,V[O_u].Size\gets cnt$\;
\end{algorithm}

The duplicate checker employs a segmented bitmap to track visited vertices, where each vertex is mapped to a unique bit. The segmented bitmap includes a set of bitmaps (i.e., segments) of fixed length $L$. Before each compaction, it allocates additional segments as needed to ensure the total number of bits is at least $n$. Let $B$ denote the segmented bitmap where $B[i][j]$ corresponds to the $j$-th bit of the $i$-th segment and $O_v$ be the offset of the destination vertex $v$. Since each vertex occupies a 32-byte block in the vertex table, the logical ID of vertex $v$ is $O_v / 32$, which is mapped into $[0,n-1]$. Accordingly, the corresponding bit of vertex $v$ is located at $B[(O_v / 32)/L)][(O_v / 32) \bmod L]$. If this bit equals 1, $v$ has been visited; otherwise, it is unvisited. We provide a complete example about the bitmap states during the compaction in Figure \ref{fig:compaction}. Initially, the bitmap has 3 segments and all bits are 0. For offset (64), its corresponding segment is $64/32/L=1$, and corresponding bit in that segment is $(64/32) \text{ mod } L=0$. Therefore, $B[1][0]$ is set as 1. Later, when traversing all blocks with offset (64), they are skipped since the corresponding bit has already been marked.

\begin{theorem}
    \label{theorem:edge-complexity}
    Under a single-threaded execution model, the amortized time complexity of edge insertion, update and deletion is $O(1)$.
\end{theorem}
\begin{proof}[Proof Sketch]
    We prove the theorem by discussing the two cases of edge operations: (1) when the log segment is not full, appending an edge log to the segment clearly costs $O(1)$ time; (2) when a compaction is required, the compaction process costs $O(d)$, and is amortized to $O(1)$ for each edge operation. The complete proof can be found in Appendix B.
\end{proof}

We remark that the concurrency overhead is low, as synchronization is only required in vertex-local compactions and does not block updates on other vertices.
\begin{figure}[t]
    \includegraphics[width=0.9\linewidth]{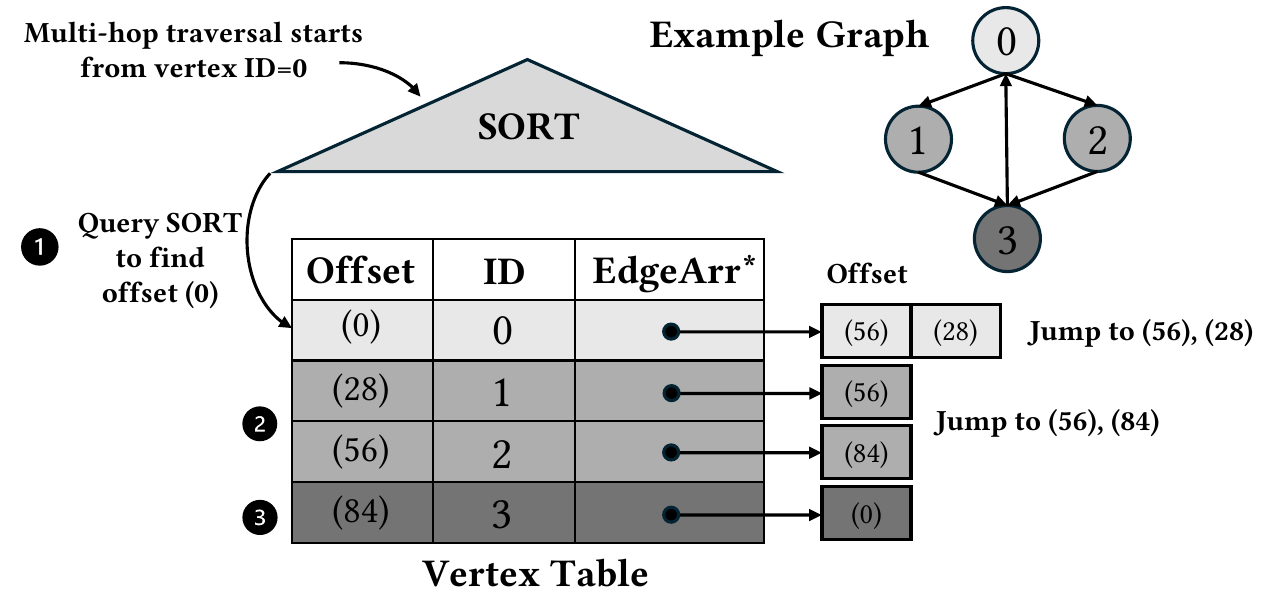}
    \caption{An example of traversing multiple hops starting from $ID=0$; Only step-\circnum{1} needs to query SORT.}
    \label{fig:logchain}
    \vspace{-1mm}
\end{figure}

\vspace{1mm}
\noindent\textbf{Edge chain.} Existing graph systems typically store the destination vertex IDs within their edge structures. Consequently, during graph traversals (e.g., multi-hop queries), each visited vertex must be looked up in the vertex index to retrieve its corresponding neighbor edges. For example, when performing a breadth-first search (BFS) \cite{adjlist} starting from vertex 0 on the graph shown in Figure~\ref{fig:logchain}, the system must perform four vertex index lookups—one for each of the visited vertices: 0, 1, 2, and 3.

In contrast, the edge array in RadixGraph stores vertex offsets instead of vertex IDs. By storing destination vertex's offset, the edges form a chain structure that directly represents the topological structure of the original graph, as shown in Figure \ref{fig:logchain}. This effectively creates a direct path from one vertex to its neighbors, allowing traversal to proceed seamlessly along this chain. For the same BFS process starting from vertex 0, RadixGraph only needs to perform a one-time lookup in the SORT to locate the starting vertex. Once the offset of the starting vertex is known, the traversal can continue through the edge chain by following the offsets embedded in the edge blocks without further vertex index accesses.
\begin{figure}[t]
    \includegraphics[width=0.45\textwidth]{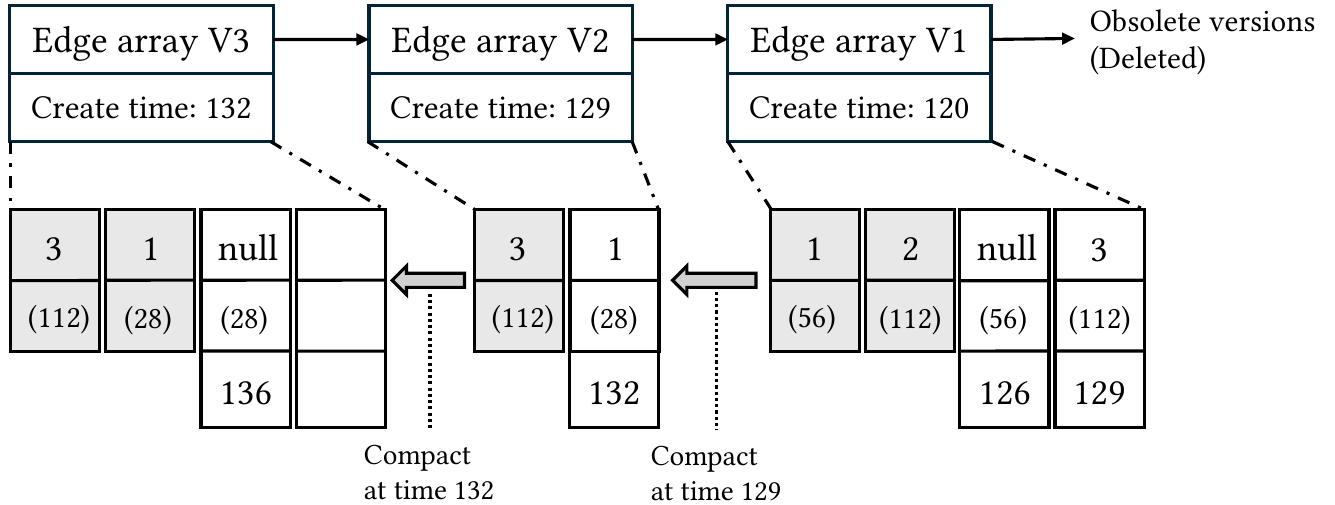}
    \caption{An example multi-versioned edge arrays.}
    \label{fig:versioning}
    \vspace{-1mm}
\end{figure}

\vspace{1mm}
\noindent\textbf{Version management.}
Many graph analytics workloads require access to a consistent snapshot of the graph. RadixGraph supports multi-version concurrency control (MVCC) \cite{mvcc} by associating timestamps with both the snapshot segment and each appended log block. Multiple versioned edge arrays are organized as a linked list, as illustrated in Figure \ref{fig:versioning}. A versioned array becomes eligible for deletion only after it has been compacted or its associated vertex has been deleted, and it is no longer accessed by any readers.  For a graph analytics workload timestamped $t$, the system first traverses the linked list to locate the most recent versioned array created before $t$, then filters the log segment to include only updates with timestamps less than or equal to $t$. This design allows multiple analytics tasks to operate concurrently on consistent, historical views of the graph.
\subsection{Discussions}
\label{sec:discuss}
\noindent\textbf{Operation complexity.} The fundamental operations in a dynamic graph system include (1) vertex-oriented operations: locate, insert or delete a vertex, (2) edge-oriented operations: insert, delete or update an edge, and (3) query operation: get neighbors of a vertex.

We compare the time complexities of RadixGraph on fundamental operations with recent state-of-the-art works in Table \ref{tab:complexity}. With $l$ layers in the SORT, all vertex-related operations can be performed in $O(l)$ time. Note that $l$ is a hyperparamter, and we set $l=O(lglg(u))$ to achieve the same complexity as Spruce's vertex index while minimizing its average space cost compared to Spruce. While GTX achieves better complexity in some vertex operations, it relies on a concurrent hashmap for indexing vertices, which becomes significantly slower as $n$ grows large. For edge-related operations, RadixGraph requires only $O(1)$ time complexity for inserting, updating, or deleting an edge  as shown in Theorem \ref{theorem:edge-complexity}. For get-neighbors operation, RadixGraph costs $O(d)$ by performing a process similar to the log compaction.

\vspace{1mm}
\noindent\textbf{Memory consumption.} Let $n$ denote the number of vertices and $m$ the number of edges. When vertex IDs are not excessively sparse, the worst-case space complexity of SORT grows linearly with $n$, say, $kn$ for some $k$. The empirical verification of this property is presented in Section~\ref{sec:case-study} ``SORT space consumption as $n$ increases'', while the detailed theoretical analysis is provided in Appendix C. The vertex table consumes at most $64n$ bytes of memory since each vertex block costs 32 bytes and the vertex table size is at most twice the number of vertices. For the edge arrays, each edge block occupies 8 bytes, so the total memory of the snapshot segments does not exceed $8m$ bytes. Since the log segment size is set equal to the snapshot segment size and each log block costs 12 bytes, the total memory of the log segments is bounded by $12m$ bytes. Each duplicate checker maintains a segment bitmap of $L\lceil\frac{n}{L}\rceil$ bits, requiring $L\lceil\frac{n}{L}\rceil/8$ bytes. Let $t$ be the maximum number of worker threads allowed for concurrent compactions (typically a constant such as 32 or 64), the total bitmap space is $t \cdot L\lceil\frac{n}{L}\rceil/8$ bytes. Therefore, the overall space required to store the graph is $O(kn+64n+8m+12m+tL\lceil\frac{n}{L}\rceil/8)=O(m)$.

\begin{table}[t]
    \centering
    \fontsize{7.8pt}{1.1em}\selectfont
    \setlength{\tabcolsep}{1.5mm}{
    \begin{tabular}{ccccc}
    \hline
    Operation & Teseo & Spruce & GTX & RadixGraph\\
    \hline
    locate\_v & $O(lg(u))$ & $O(lglg(u))$ & $O(1)$ & $O(lglg(u))$ \\
    insert\_v & $O(lg(u))$ & $O(lglg(u))$ & $O(1)$ & $O(lglg(u))$ \\
    delete\_v & $O(lg(u)+d)$ & $O(lglg(u)+d)$ & $O(d^2)$ & $O(lglg(u)+d)$ \\
    insert\_e & $O(lg^2(d))$ & $O(d)$ & $O(1)$ & $O(1)$ \\
    delete\_e & $O(lg^2(d))$ & $O(lg(d))$ & $O(d)$ & $O(1)$ \\
    update\_e & $O(lg^2(d))$ & $O(lg(d))$ & $O(d) $& $O(1)$ \\
    get\_ngbrs & $O(d)$ & $O(d)$ & $O(d)$ & $O(d)$ \\
    \hline
    \end{tabular}}
    \caption{Time complexity of fundamental operations. $n$ is the number of vertices, $u$ is the size of the range of vertex identifiers and $d$ is the average degree of vertices. ``get\_ngbrs'' represents the get-neighbors operation.}
    \label{tab:complexity}
    \vspace{-1mm}
\end{table}

\vspace{1mm}
\noindent\textbf{Concurrency control.} Concurrency control is another critical factor in the performance of a graph system. Our vertex index concurrency protocol builds on Read-Optimized Write EXclusion (ROWEX) \cite{rowex, spruce}. RadixGraph employs an atomic bitmap in each internal node of the SORT, where an $i$-layer node maintains a bitmap of size $2^{a_i}$. Each bit in the bitmap uniquely corresponds to a child node. When a new child node is created, the corresponding bit in its parent node is atomically set to 1 using an atomic compare-and-swap (CAS) operation \cite{10.5555/2385452}, indicating an ongoing creation process, and is reset to 0 once the creation completes. For readers, they simply read the data atomically. To support concurrent edge operations, RadixGraph maintains an atomic \textit{Size} field in its vertex block, and each log entry obtains a unique identifier by an atomic add operation on \textit{Size}. For readers, they can perform latch-free read on the edge array.

\vspace{1mm}
\noindent\textbf{Supporting graph types.} RadixGraph is primarily designed for directed weighted graphs. However, the ``Weight'' field in edge and log blocks can be replaced with other properties to support other types of graphs like \textit{temporal graphs} \cite{temporal_graph} and \textit{labelled graphs} \cite{labelled_graph} as long as a ``NULL'' value is reserved for the field. RadixGraph can also be easily extended to support \textit{multi-graphs}. In that case, each vertex block should maintain a counter to represent the number of edges inserted in its edge array, and each edge and log block should maintain an extra ``ID'' field. Upon inserting a new edge to that vertex, the edge gets a unique ID by incrementing the counter. During log compactions, edge and log blocks with the same IDs will be compacted and only the latest version is kept. If maintaining \textit{vertex labels} \cite{bonifati2018querying} with MVCC consistency is required, we can also extend the vertex block to store a pointer referring the versioned label chain of the vertex.
\section{Experiments}
\label{lab:exp}
\subsection{Setup}
\label{sec:setup}
\textbf{System environment.} We perform our experiments on a dual-socket machine equipped with Intel(R) Xeon(R) Gold 6326 CPU @ 2.90GHz processors and 250GB RAM. Each CPU has 48MB L3 Cache, 16 cores, and supports at most 32 threads. Unless otherwise specified, all experiments are executed with 64 threads running concurrently. All the codes were compiled with GCC 11.4.0 on Ubuntu Linux with O3 optimization.

\vspace{1mm}
\noindent\textbf{Graph data.} We use various graph datasets from SNAP \cite{snapnets} and LDBC graph analytics \cite{ldbc}, as shown in Table \ref{tab:datasets}. Two synthetic datasets include (1) \textit{g500-24}, power-law graphs, and (2) \textit{uni-24}, uniform-law graphs. These datasets provide diverse structural properties under different graph topologies. Four real-world datasets include (1) \textit{livejournal}, the LiveJournal social network \cite{livejournal}; (2) \textit{dota}, the relationships between game entities \cite{dota-league}; (3) \textit{orkut}, a relationship graph capturing ground-truth communities from the Orkut social network \cite{orkut}; and (4) \textit{twitter-2010}, Twitter follower network as of 2010 \cite{twitter2010}. Note that all the graphs are treated as undirected graphs, and each edge would be inserted in both directions.

\begin{table}[t]
    \fontsize{8pt}{1.15em}\selectfont
    \setlength{\tabcolsep}{1.6mm}{
    \begin{tabular}{l|c|c|c|c}
    \hline
    \textbf{Datasets} & \#\textbf{Vertices} & \#\textbf{Edges} & \textbf{Avg. Deg.} & \textbf{Max. Deg.}\\
    \hline\hline
    livejournal (lj) & 4M & 34M & 17.35 & 14815\\
    \hline
    dota & 61K & 51M & 1663.24 & 17004\\
    \hline
    orkut & 3M & 117M & 76.28 & 33313\\
    \hline
    g500-24 (g24) & 9M & 260M & 58.70 & 406416\\
    \hline
    uni-24 (u24) & 9M & 260M & 58.70 & 103\\
    \hline
    twitter-2010 (twitter) & 41M & 1.46B & 70.51 & 3081112\\
    \hline
    \end{tabular}
    }
    \caption{Datasets used in our experiments. ``K'', ``M'', ``B'' represent ``thousand'', ``million'' and ``billion''.}
    \label{tab:datasets}
    \vspace{-1mm}
\end{table}

\vspace{1mm}
\noindent\textbf{Graph benchmark.} Recently, some benchmarks are developed to evaluate in-memory dynamic graph systems \cite{10.1145/3709720,teseo,gapbs,ldbc}. We evaluate RadixGraph and its competitors based on GFE (Graph Framework Evaluation) driver \cite{teseo}, a C++ driver to evaluate updates and analytics in dynamic structural graphs. For graph analytics, we implement parallel $k$-hop neighbor queries and graph algorithms with GAPBS (GAP Benchmark Suite) \cite{gapbs}, including breadth-first search (BFS), single-source shortest paths (SSSP), PageRank (PR), weakly connected components (WCC), triangle counting (TC) and betweeness centraility (BC).

\vspace{1mm}
\noindent\textbf{Competitors.} We choose recent state-of-the-art in-memory dynamic graph systems for comparisons: Teseo \cite{teseo}, Sortledton \cite{sortledton}, Spruce \cite{spruce} and GTX \cite{gtx}, all implemented in C++ with parallelism support. We also compare RadixGraph with Terrace \cite{terrace}, Aspen \cite{aspen} and CPAM \cite{cpam}, which are optimized for batch updates. Note that some methods failed to complete on some datasets either due to exceptions or not finishing within 24 hours, and hence omitted from the figures. We also note that these baselines may have different levels of transactional supports which correlates with operational performance or memory consumption. For example, GTX has a stronger transactional guarantee with snapshot isolation and serializable transactional support.

\begin{figure*}[t]
    \centering
    \includegraphics[width=1.\linewidth]{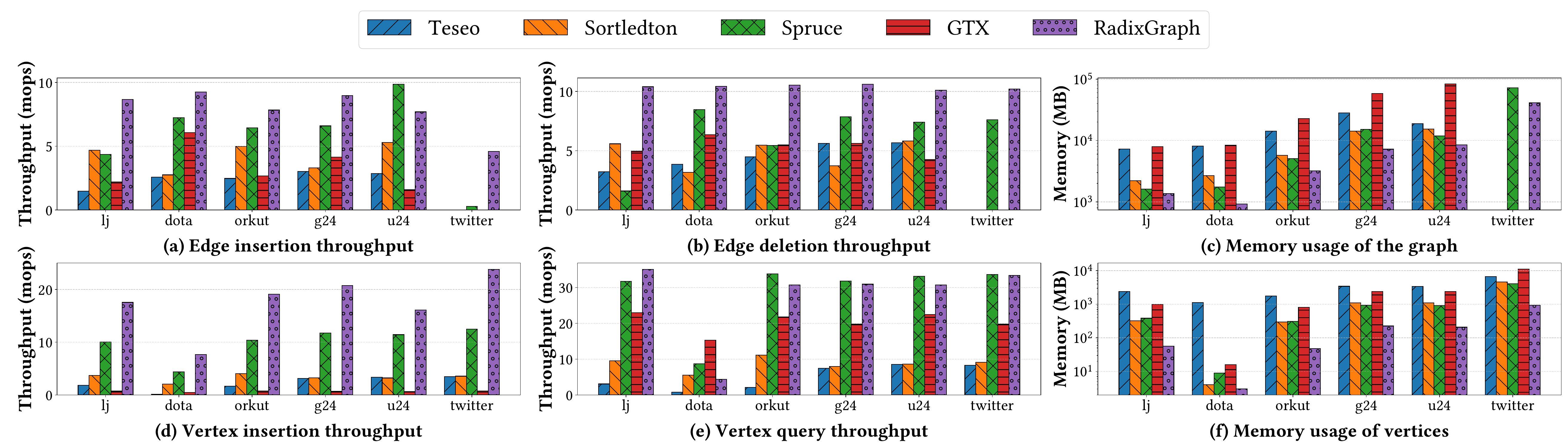}
    \caption{Throughput of edge and vertex operations and memory consumptions.}
    \label{fig:dynamic-ops}
    \vspace{-2mm}
\end{figure*}
\begin{figure}[t]
    \centering
    \includegraphics[width=1.\linewidth]{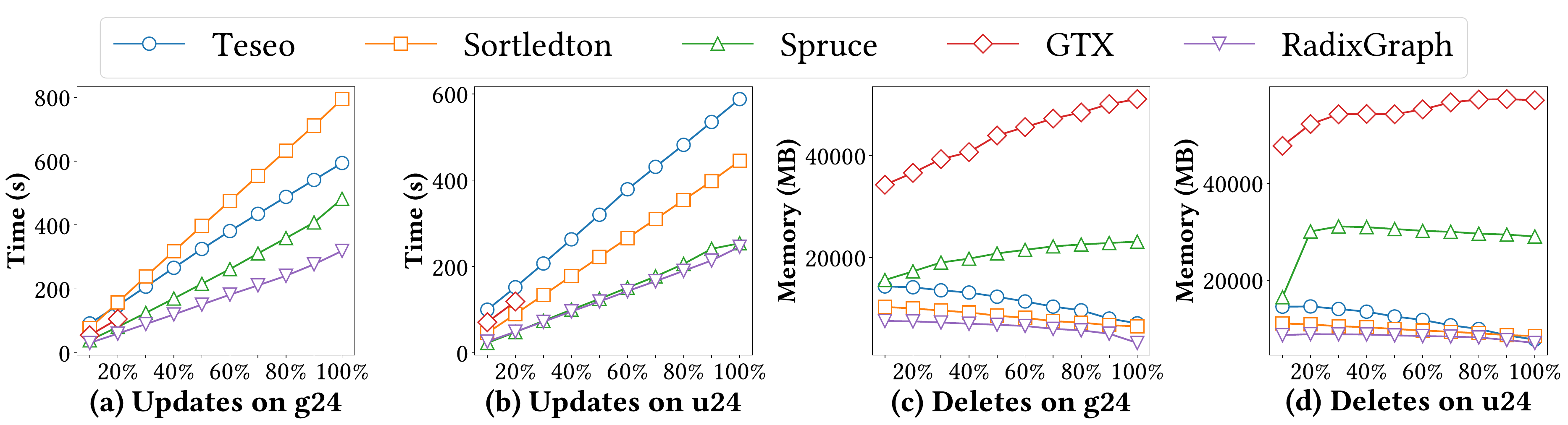}
    \caption{The time cost of mixed edge updates and memory footprint of large-scale edge deletions for all methods.}
    \label{fig:time-and-memory-footprint}
    \vspace{-1mm}
\end{figure}
\subsection{Dynamic operations}
\label{sec:dynamic}
\noindent\textbf{Edge insertion throughput.} Figure~\ref{fig:dynamic-ops}(a) shows the insertion throughput of all methods across all datasets. On the twitter dataset, Sortledton, Teseo and GTX fail to finish in our experimental environment and thus are omitted. In this experiment, edges are randomly permuted and inserted in a round-robin manner and vertices are inserted alongside if they do not exist in the system. RadixGraph outperforms all other systems on most datasets, achieving up to $16.27\times$ higher throughput than Spruce on the twitter dataset. Although Spruce performs best in the uniform graph (u24), its throughput degrades significantly on the twitter dataset. This is because Spruce appends neighbor edges to a fixed-size buffer, which is merged into a sorted array only when full. While this strategy yields near-constant insertion time when merges are infrequent, high-degree vertices in the twitter graph (up to 3 million neighbors) trigger frequent merges, resulting in degraded $O(d)$ insertion complexity. GTX employs a delta-chain index per vertex to enhance concurrency and provide transactional support. This design is beneficial for dense graphs like dota, where concurrent writes to the same vertex are common, allowing GTX to outperform Sortledton and Teseo, which rely on per-vertex latches. However, in sparser graphs with lower average degree, the delta-chain is underutilized, and transactional support inherently requires trade-offs with performance. Overall, RadixGraph generally has a good performance on different types of datasets thanks to its superior complexity, compact edge structure and latch-free log append process.

\vspace{1mm}
\noindent\textbf{Edge deletion throughput.} Figure~\ref{fig:dynamic-ops}(b) presents the deletion throughput of all methods across all graphs. This experiment directly follows the insertion benchmark, with edges deleted in the same round-robin order as they were inserted. We focus on evaluating the performance of the deletion phase in isolation. RadixGraph achieves the highest deletion throughput across all datasets, outperforming the second-best system by $1.23\times$-$1.93\times$ on the datasets, respectively. In general, we observe that most graph systems achieve higher throughput for deletions than for insertions. For Spruce, this is due to its deletion complexity being $O(lg(d))$, compared to its insertion complexity of $O(d)$. For the other systems, the primary reason is that edge deletions do not require writing edge properties, reducing the I/O and processing overhead.

\vspace{1mm}
\noindent\textbf{Vertex operations.} Figures~\ref{fig:dynamic-ops}(d) and~\ref{fig:dynamic-ops}(e) show the vertex insertion and query throughput of all methods. RadixGraph achieves the highest vertex insertion performance across all datasets, with Spruce being the second best. This advantage comes from the radix-tree-like structure, which avoids costly split, merge, or resize operations during updates. Compared with Spruce which employs a vEB-tree as its vertex index, RadixGraph benefits from its space-optimization model that minimizes memory consumption. The reduced space usage leads to fewer pointer allocations, thereby accelerating insertions. For vertex queries, RadixGraph and Spruce also outperform other methods in most cases, owing to their superior $O(lglg(u))$ query complexity. Although multi-level vector structures have a theoretical $O(1)$ query time, they require maintaining an external hash map for ID translation. Each query must first access the hash map before traversing multiple vector levels, and this complex design introduces overhead that ultimately degrades throughput. GTX has the unique benefit to support transactional and serializable graph operations; as a result, the additional coordination and validation required by its concurrency control mechanisms introduce extra overhead during vertex insertions, naturally leading to lower throughput even though it employs a similar vertex index as Sortledton’s.

\vspace{1mm}
\noindent\textbf{Mixed edge updates.} Figure~\ref{fig:time-and-memory-footprint}(a) and (b) show the update time footprint of all methods on synthetic datasets. This experiment follows the same setting as tested in the Teseo and Sortledton papers \cite{teseo,sortledton}. GTX encounters an OOM exception on both datasets when processing $20\%$ operations, and thus its curve contains only two data points. On g500-24, RadixGraph consistently outperforms the other methods. On uni-24, RadixGraph and Spruce exhibit similar performances for low vertex degrees. RadixGraph's curves show no latency spikes, demonstrating stable $O(1)$ amortized complexity and the log compactions do not significantly influence the efficiency.

\begin{figure*}[t]
    \centering
    \includegraphics[width=0.98\linewidth]{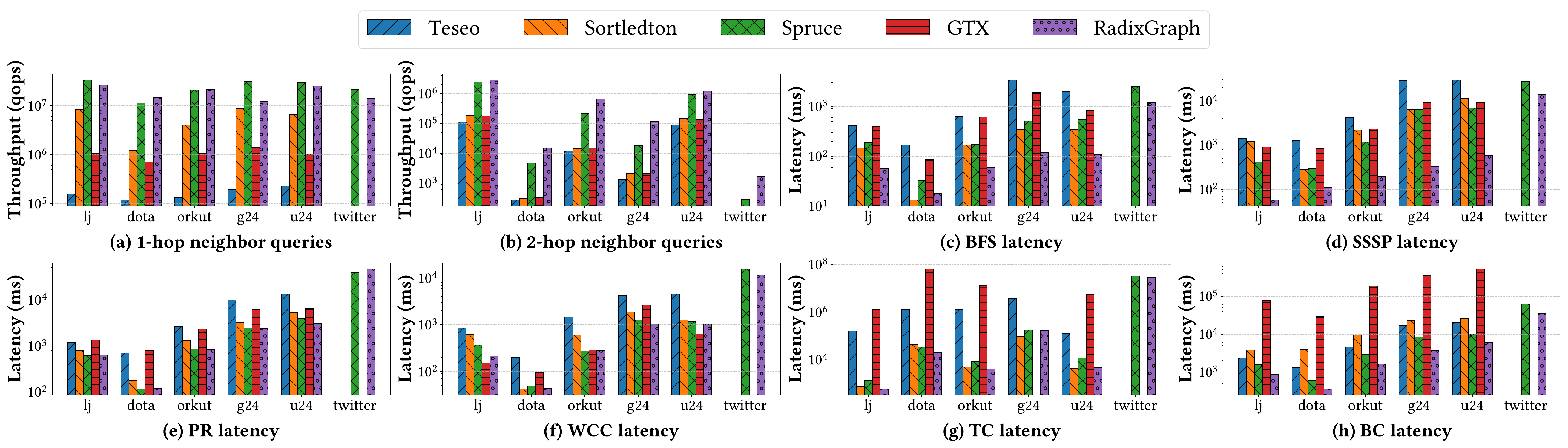}
\caption{Throughput of hop queries and latency of graph algorithms.}
\label{fig:graph-algorithms}
\vspace{-2mm}
\end{figure*}
\subsection{Memory consumption}
\label{sec:memory}
\textbf{Storing the whole graph.} Figure \ref{fig:dynamic-ops}(c) shows the physical memory usage of storing dynamic graphs by edge insertions across different methods. On the twitter dataset, only Spruce and RadixGraph construct the graph successfully so we only report their results. We measure the memory difference before and after graph construction. RadixGraph demonstrates the lowest memory consumption on all datasets, using an average of $40.1\%$ less space than the second-lowest memory consumption. Although Spruce also employs a radix-tree-based structure (i.e., vEB-tree) for its vertex index, its structure remains fixed regardless of graph size, whereas RadixGraph optimizes space efficiency through our proposed integer programming model. GTX consumes the largest space on most datasets, since it uses 64-byte edge deltas and a delta-chain index per vertex. This exchanges the memory consumption with better transactional supports.

\vspace{1mm}
\noindent\textbf{Storing vertices.} Figure~\ref{fig:dynamic-ops}(f) reports the physical memory usage for storing only the vertices of each dataset. RadixGraph achieves the lowest memory consumption across all datasets. In contrast, Teseo exhibits consistently high memory usage regardless of graph size, as it stores both vertices and edges within the leaf nodes of its ART. Consequently, inserting a vertex immediately triggers the allocation of edge segments, even when no edges exist. Although Spruce also employs a radix-tree-like structure as its vertex index, its space efficiency is only comparable to that of Sortledton and does not show significant advantages. This is because Spruce’s radix tree is not optimized for minimal space usage, unlike RadixGraph’s structure, which is tuned through its optimization model.

\vspace{1mm}
\noindent\textbf{Memory footprints of edge deletions.} Figures~\ref{fig:time-and-memory-footprint}(c) and (d) show the memory footprints during edge deletions on synthetic datasets. The memory usage of RadixGraph gradually decreases, with a more pronounced drop in the later stages of deletion. This behavior arises because in RadixGraph, the log segment size equals the snapshot segment size. Therefore, an edge array is compacted and recycled only after all its edges have been deleted (i.e., the log segment is fully filled). This observation suggests that, for delete-heavy workloads, reducing the log segment size relative to the snapshot segment size could enable more aggressive compaction and lower space overhead. Spruce initializes a linked list of versioned edges at the start of deletions and records each deleted edge in this list without employing any garbage collection mechanism. As a result, its memory footprint may even increase during edge deletions. GTX exhibits similar behavior, as its lazy garbage collection delays memory reclamation to reduce runtime overhead, resulting in temporary memory growth under delete-heavy workloads.

\begin{figure}[t]
    \centering
    \includegraphics[width=1.\linewidth]{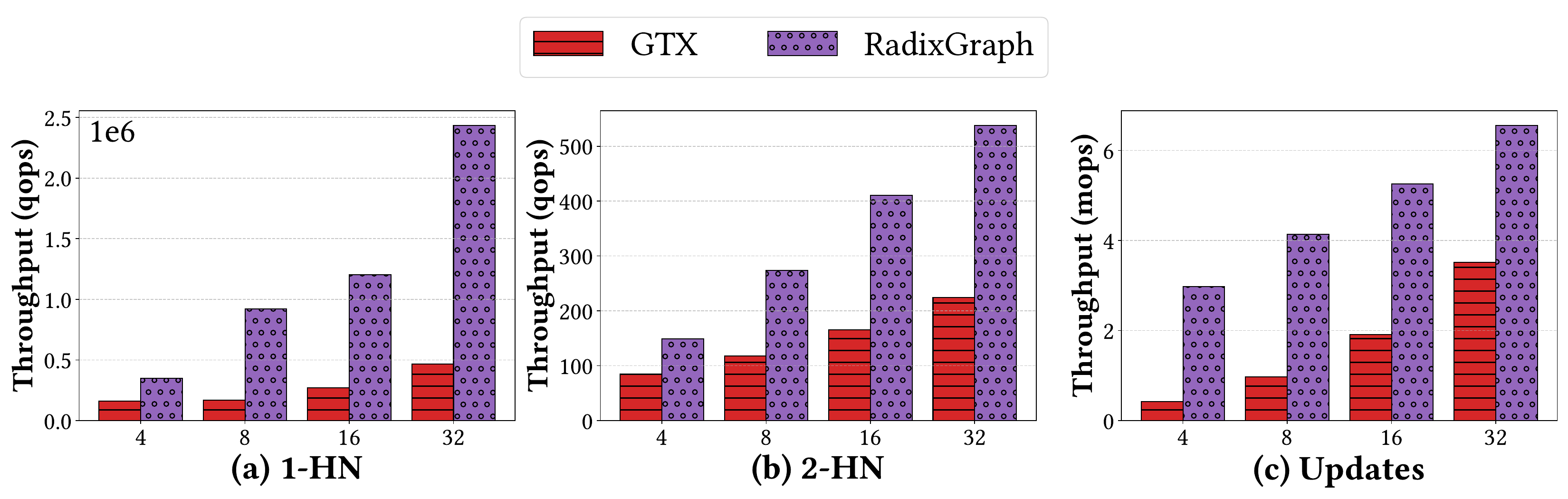}
    \caption{The throughput of concurrent reads and writes on \textit{dota} dataset over different number of threads. (a) refers to 1-hop neighbor query throughput; (b) refers to 2-hop neighbor query throughput; (c) refers to graph update throughput.}
    \label{fig:concurrent-throughput}
    \vspace{-1mm}
\end{figure}
\subsection{Graph analytics}
\label{sec:analytics}
In this subsection, we evaluate the efficiency of graph analytical tasks across different methods. Note that on the Twitter dataset, only Spruce and RadixGraph construct the graph and execute graph analytics successfully so we only report their results.

\vspace{1mm}
\noindent\textbf{Neighbor queries.} Figure~\ref{fig:graph-algorithms}(a) and (b) show the throughput of $k$-hop neighbor queries. The experiment first constructs the graph by inserting all edges, and then executes $k$-hop neighbor queries from every vertex in parallel. Our evaluation focuses on measuring the query throughput. For 1-hop queries, RadixGraph and Spruce achieve the highest performance due to their compact edge array designs, which enable efficient sequential scans for neighbor retrieval. For 2-hop queries, RadixGraph outperforms Spruce by up to $6.11\times$, owing to its edge chain structure that avoids redundant vertex lookups when traversing the neighbors of 1-hop neighbors. Note that the time complexity of a 2-hop query is $O(d^2)$, so its throughput decreases on dense and power-law graphs (e.g., dota, g24, twitter) compared to uniform-law graphs (e.g., u24).

\vspace{1mm}
\noindent\textbf{Graph algorithms.} We implement two parallel single-source traversal algorithms (BFS and SSSP) and four parallel iterative graph algorithms (PR, WCC, TC and BC) based on the GAP Benchmark Suite (GAPBS). Single-source traversal algorithms start from a given vertex and explore its reachable subgraph, while iterative algorithms perform repeated computations over the entire graph. Figures~\ref{fig:graph-algorithms}(c) and (d) present the performance of each method on BFS and SSSP using a logarithmic scale. RadixGraph significantly reduces BFS and SSSP execution time on most datasets. This improvement is primarily due to RadixGraph’s edge chain design, which enables direct traversal from a vertex without repetitive vertex index lookups. In contrast, prior methods require vertex index lookups at each step. Figures~\ref{fig:graph-algorithms}(e), (f), (g) and (h) show the performance of the systems on PR, WCC, TC and BC. The complexity of TC algorithm is much higher than other algorithms and GTX fails to finish within 24 hours on g24 dataset since its stronger transactional support degrades performance. For PR and WCC, although RadixGraph does not benefit from the same lookup avoidance in these algorithms (as they need to iterate the whole graph), it still delivers competitive performance across all datasets due to its compact edge array design. BC involves BFS executions in its process, so RadixGraph also benefits from the BFS speedup and outperforms other baselines. For TC, although it involves two-hop neighbor queries, the dominant cost lies in computing neighbor intersections rather than neighbor retrievals, so RadixGraph gains limited benefits in this case.
\subsection{Concurrent reads and writes}
\label{sec:concurrent}
Real-world workloads often involve concurrent reads and writes. RadixGraph supports MVCC that allows read transactions to proceed without blocking or conflicting with concurrent write transactions. To evaluate performance under such workloads, we generate update operations following the procedure described in the mixed updates experiment and execute them concurrently with neighbor queries. For comparison, we include GTX, a state-of-the-art transactional graph system that also supports MVCC. Other systems are excluded from this experiment, as they either encounter deadlocks or segmentation faults under concurrent workloads, as also reported in the GTX paper \cite{gtx}. Figure \ref{fig:concurrent-throughput}(a) and (b) show the concurrent read throughput for 1-hop and 2-hop neighbor queries, respectively, with a varying number of readers and fixed 32 writers. Figure \ref{fig:concurrent-throughput}(c) presents the concurrent write throughput when varying the number of writers while holding the number of readers constant at 32. RadixGraph shows strong scalability for both concurrent read efficiency and write efficiency.

\subsection{Case study of SORT}
\label{sec:case-study}
\noindent\textbf{Robustness of SORT over $n$.} We examine how the number of vertices $n$ affects the optimal fanout configuration and memory consumption of SORT. We fix $u = 2^{32}$, meaning vertex IDs are distributed within the range $[0, 2^{32} - 1]$, and set the number of layers in SORT to $l = lglg(u)= 5$. Figure~\ref{fig:fanout-evolve} shows how the optimal fanouts evolve as $n$ increases. We observe that the changes occur at exponentially spaced intervals of $n$, with only five configuration shifts between $10^5$ and $10^7$. This suggests that the fanout can remain fixed over a broad range of $n$ and only needs to be updated when $n$ grows by roughly an order of magnitude. Additionally, the fanout values $a_i$ vary only slightly across different $n$, indicating that the optimal configuration is stable and robust. For example, when $n$ changes from $50000$ to $300000$, the optimal configuration changes from $\left\{19,4,3,3,3\right\}$ to $\left\{20,3,3,3,3\right\}$. Figure~\ref{fig:memory-evolve} presents the memory usage across different $n$ values under various radix tree configurations. The ``Updated'' bars represent the most recently computed optimal configuration, while the ``Trailing'' bars reflect the previously optimal configuration that has not yet been updated. The results show that using the updated configuration yields approximately $5\%$ lower memory consumption than the trailing one. Nonetheless, the difference is modest, and both configurations consistently outperform the vEB-tree baseline.
\begin{figure}[t]
    \vspace{-3mm}
    \subfigure[SORT fanouts] {
        \includegraphics[width=0.24\linewidth]{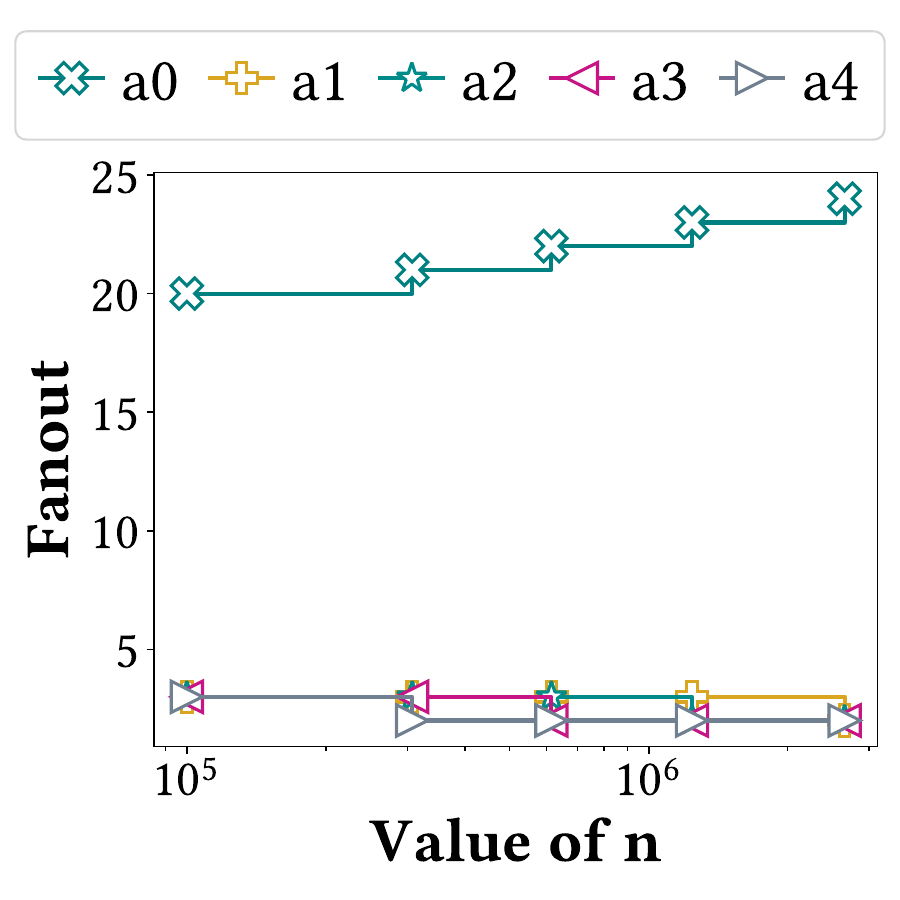}
        \label{fig:fanout-evolve}
    }
    \hspace{-2mm}
    \subfigure[Mem-costs] {
        \includegraphics[width=0.24\linewidth]{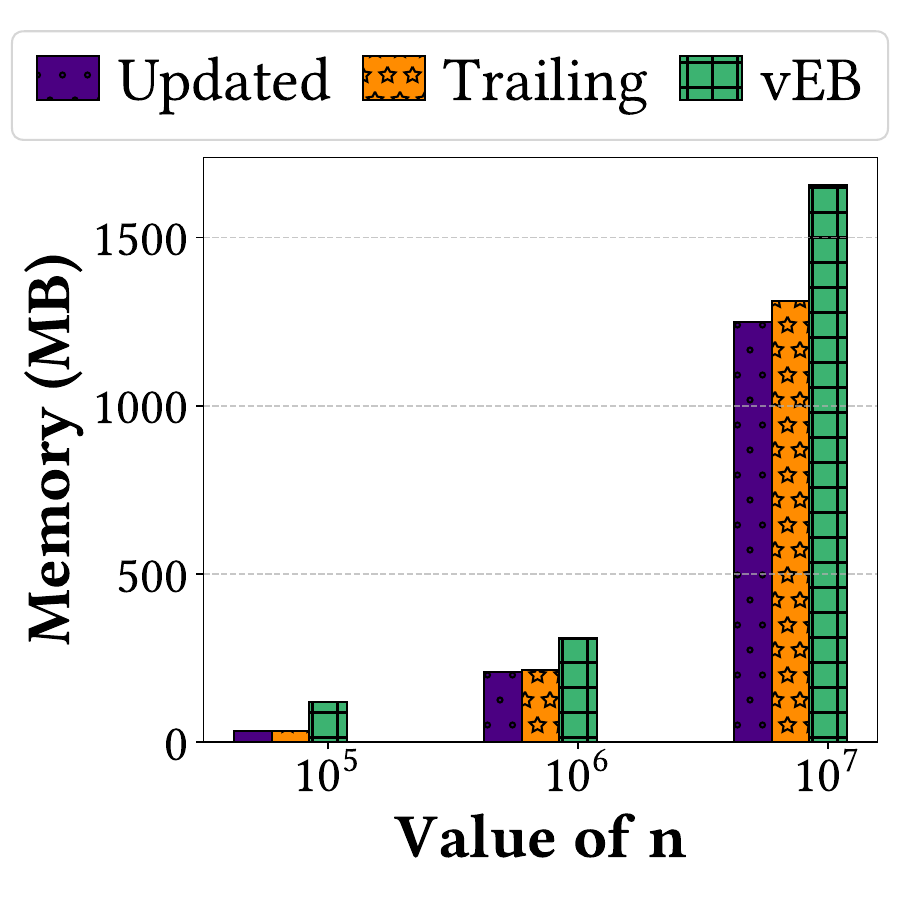}
        \label{fig:memory-evolve}
    }
    \hspace{-2mm}
    \subfigure[Mem-footprint]{
        \includegraphics[width=0.24\linewidth]{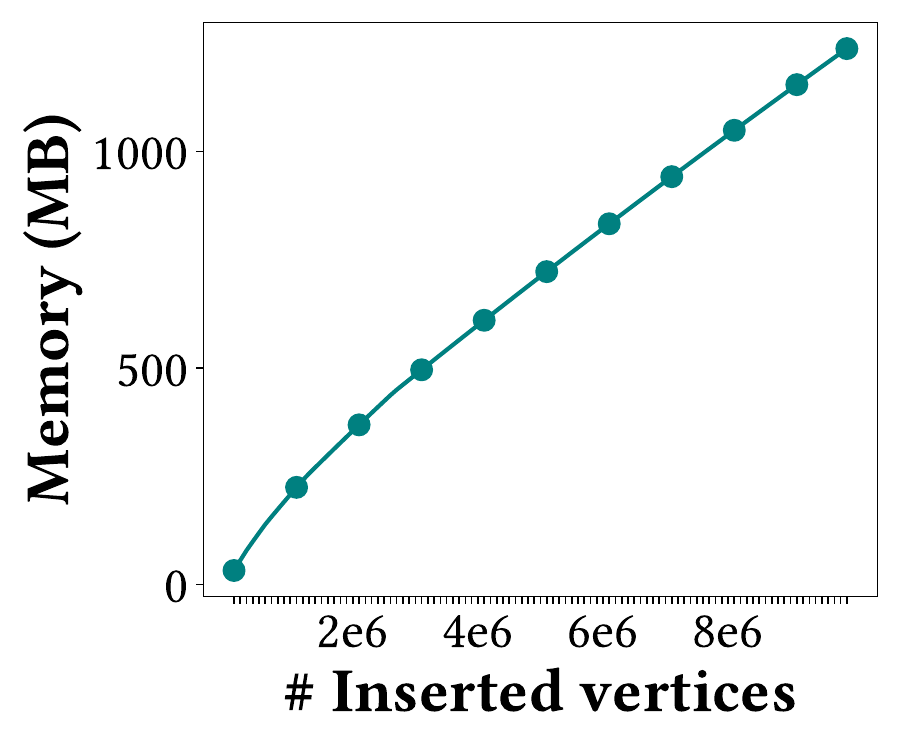}
    }
    \hspace{-2mm}
    \subfigure[Adapt-costs]{
    \begin{minipage}[b]{0.2\linewidth}
    \raisebox{10mm}{
    \fontsize{6pt}{0.9em}\selectfont
    \begin{tabular}{cc}
    \hline
    $n$ & Time\\
    \hline
    3.1E5 & 0.23s \\
    6.2E5 & 0.46s \\
    1.2E6 & 0.93s \\
    2.7E6 & 1.96s \\
    \hline
    \end{tabular}}
    \end{minipage}
    }
    \caption{The (a) optimal configurations, (b) memory consumptions for different $n$; and (c) the memory footprint, (d) the transformation costs  during insertions.}
    \label{fig:sort-evolve}
    \vspace{-1mm}
\end{figure}

\vspace{1mm}
\noindent\textbf{Transformation cost of SORT.} We evaluate the time overhead of continuously transforming SORT as the number of vertices $n$ increases from $10^5$ to $10^7$. Figure~\ref{fig:sort-evolve}(d) summarizes the results, where $n$ represents the exact point where the optimal configuration changes. As shown in the table, configuration updates occur approximately when $n$ doubles, and each transformation completes within a few seconds. This confirms that SORT’s reconfiguration overhead has negligible impact on overall system performance.


\vspace{1mm}
\noindent\textbf{SORT space consumption as $n$ increases.}
We further examine whether the space consumption of SORT scales linearly with the number of vertices $n$. This question is crucial, as radix trees generally lack a deterministic worst-case bound with respect to $n$. To investigate, we vary $n$ within the range $[10^5, 10^7]$, sampling 100 evenly spaced points, and record the corresponding memory usage. Figure~\ref{fig:sort-evolve}(c) demonstrates the results which show an approximately linear growth in memory consumption with $n$. This indicates that SORT’s practical space complexity is close to $O(n)$. We also provide a theoretical analysis of SORT’s worst-case space complexity, which confirms that its space remains near-linear in $n$ when vertex IDs are not excessively sparse. Due to space constraints, the detailed analysis is included in Appendix~C.

\begin{figure}[t]
    \centering
    \includegraphics[width=1.\linewidth]{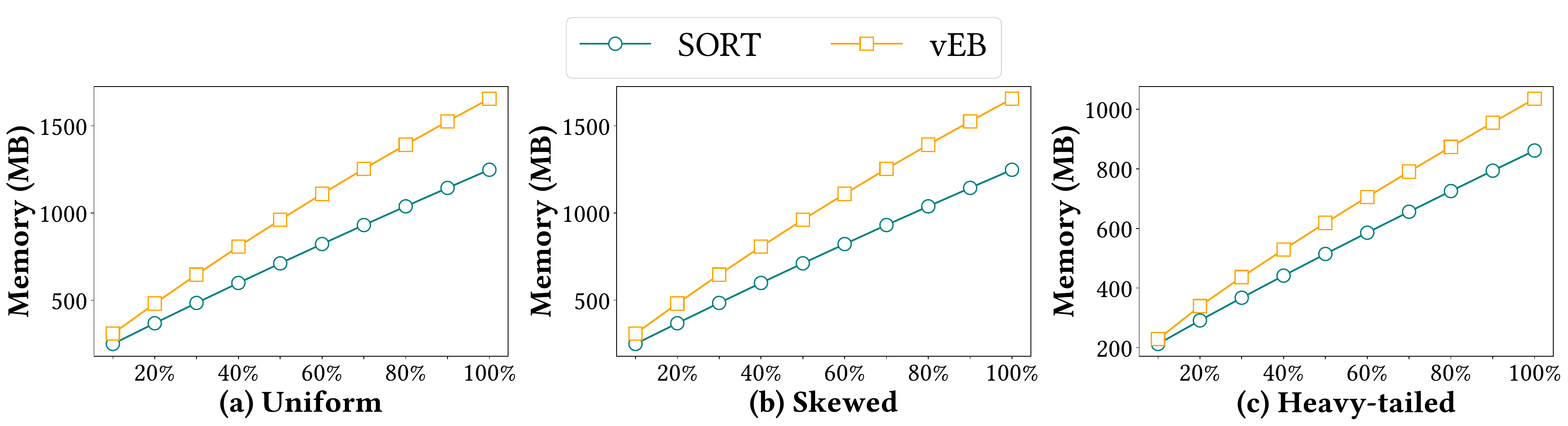}
    \caption{The memory footprints of inserting $10^7$ IDs into SORT and vEB-tree under different workloads.}
    \label{fig:workload-log}
    \vspace{-1mm}
\end{figure}
\vspace{1mm}
\noindent\textbf{SORT performance under non-uniform workloads.} We evaluate the memory efficiency of SORT and the vEB-tree under three types of workloads: (1) \textit{Uniform}: IDs are uniformly distributed within $[0, 2^{32}-1]$, which matches our theoretical analysis; (2) \textit{Skewed}: IDs are uniformly distributed within $[0, \frac{2^{32}-1}{1.5}]$, making the most significant bit more likely to be 0; (3) \textit{Heavy-tailed}: IDs follow a reciprocal distribution \cite{hamming1970distribution} over $[0, 2^{32}-1]$ (i.e., $p(i)\propto\frac{1}{i}$), where smaller IDs occur much more frequently. For each workload, we generate $10^7$ distinct IDs and insert them into SORT and vEB-tree, respectively. Figure~\ref{fig:workload-log} presents the memory footprints during insertions. SORT consistently outperforms the vEB-tree across all workloads, though the performance gains decrease as the distribution becomes more skewed.

\begin{table}[t]
\centering
\centering
\fontsize{6.3pt}{1em}\selectfont
\setlength{\tabcolsep}{0.7mm}{
\begin{tabular}{l|cc|cc|cc|cc|cc|cc}
\hline
\multirow{3}{*}{$n$} & \multicolumn{4}{c|}{Insertion} & \multicolumn{4}{c|}{Query} & \multicolumn{4}{c}{Memory (KB)} \\ \cline{2-13}
& \multicolumn{2}{c|}{$u=2^{24}$} & \multicolumn{2}{c|}{$u=2^{32}$} & \multicolumn{2}{c|}{$u=2^{24}$} & \multicolumn{2}{c|}{$u=2^{32}$} & \multicolumn{2}{c|}{$u=2^{24}$} & \multicolumn{2}{c}{$u=2^{32}$} \\ \cline{2-13}
& SORT & ART & SORT & ART & SORT & ART & SORT & ART & SORT & ART & SORT & ART \\ \hline
$10^4$ & \textbf{\underline{2.5E5}} & 1.7E4 & 6.0E4 & \textbf{\underline{1.7E5}} & \textbf{\underline{1.9E7}} & 1.2E6 & \textbf{\underline{6.6E6}} & 1.2E6 & 2.6E3 & \textbf{\underline{5.1E2}} & 3.1E3 & \textbf{\underline{5.1E2}} \\
$10^5$ & \textbf{\underline{8.6E5}} & 1.7E5 & \textbf{\underline{2.9E5}} & 1.8E5 & \textbf{\underline{1.2E8}} & 2.4E6 & \textbf{\underline{7.1E7}} & 2.5E6 & 4.5E4 & \textbf{\underline{1.6E4}} & 6.7E4 & \textbf{\underline{1.6E4}} \\
$10^6$ & \textbf{\underline{7.4E5}} & 1.5E5 & \textbf{\underline{5.0E5}} & 1.7E5  & \textbf{\underline{1.4E8}} & 3.5E6 & \textbf{\underline{7.0E7}} & 3.2E6 & \textbf{\underline{1.6E5}} & 2.7E5 & 6.4E5 & \textbf{\underline{2.0E5}} \\
$10^7$ & \textbf{\underline{2.6E6}} & 1.4E5 & \textbf{\underline{1.7E6}} & 1.1E5 & \textbf{\underline{7.7E8}} & 3.7E6 & \textbf{\underline{9.8E7}} & 4.2E6 & \textbf{\underline{1.6E5}} & 2.7E6 & \textbf{\underline{1.3E6}} & 2.9E6 \\
\hline
\end{tabular}}
\caption{Insertion, query throughput and memory consumption of SORT and ART under different $n$ and $u$.}
\label{tab:sort-art}
\vspace{-1mm}
\end{table}

\vspace{1mm}
\noindent\textbf{Comparing SORT and ART.} Table~\ref{tab:sort-art} presents the insertion, query throughput and memory consumption of SORT and ART under varying configurations, where $n$ denotes the number of inserted elements and $u$ represents the size of the ID universe. For ART, we adopt unodb\footnote{https://github.com/laurynas-biveinis/unodb}
, an implementation that follows the designs proposed in prior work for both sequential and parallel settings~\cite{ART,rowex}. We evaluate the parallel version. Overall, SORT achieves higher throughput and better scalability than ART. In terms of memory efficiency, ART performs better when the ID universe is sparse (large $\frac{u}{n}$), whereas SORT is more efficient when the universe is dense (small $\frac{u}{n}$). This is because in sparse universes, even an optimized SORT configuration leaves many array slots unoccupied. These observations suggest that incorporating adaptive strategies into SORT could further improve its efficiency under sparse conditions.

\subsection{Ablation study for RadixGraph}
\label{sec:ablation}
\begin{table}[t]
\centering
\centering
\fontsize{7.5pt}{1em}\selectfont
\setlength{\tabcolsep}{1.5mm}{
\begin{tabular}{l|ccc|cccc}
\hline
\multirow{2}{*}{\textbf{Graphs}} & \multicolumn{3}{c|}{\textbf{ART v.s. SORT}} & \multicolumn{4}{c}{\textbf{Slowdown w/o edge chain}} \\
\cline{2-8}
& \textbf{Insert} & \textbf{Delete} & \textbf{Memory} & \textbf{2-hop} & \textbf{BFS} & \textbf{SSSP} & \textbf{BC} \\
\hline
lj & $8.58\times$ & $32.4\times$ & $1.86\times$ & $1.78\times$ & $27.67\times$ & $5.90\times$ & $1.56\times$ \\
orkut & $18.89\times$ & $28.33\times$ & $1.24\times$ & $1.72\times$ & $1.36\times$ & $1.14\times$ & $1.37\times$ \\
dota & $55.91\times$ & $23.74\times$ & $1.01\times$ & $1.04\times$ & $1.01\times$ & $1.36\times$ & $1.33\times$ \\
g24 & $7.44\times$ & $5.50\times$ & $1.16\times$ & $1.05\times$ & $2.33\times$ & $1.00\times$ & $1.37\times$ \\
u24 & $10.99\times$ & $30.72\times$ & $1.15\times$ & $2.10\times$ & $1.43\times$ & $1.19\times$ & $1.27\times$ \\
twitter & $12.81\times$ & $1.91\times$ & $1.30\times$ & $1.32\times$ & $1.73\times$ & $1.07\times$ & $1.12\times$ \\
\hline
\end{tabular}}
\caption{Ablation study of RadixGraph. Left part is the relative time and memory costs of ART compared with SORT.}
\label{tab:ablation}
\vspace{-1mm}
\end{table}
\begin{table*}[t]
\centering
\fontsize{7.5pt}{1.2em}\selectfont
\setlength{\tabcolsep}{4pt}
\begin{tabular}{cc|cccc|cccc|cccc}
\hline
 & \multirow{2}{*}{\textbf{Batch}} & 
\multicolumn{4}{c|}{\textbf{LJ}} & 
\multicolumn{4}{c|}{\textbf{Orkut}} & 
\multicolumn{4}{c}{\textbf{Twitter}}\\
\cline{3-14}
& & \textbf{RadixGraph} & \textbf{Terrace} & \textbf{Aspen} & \textbf{CPAM} &
   \textbf{RadixGraph} & \textbf{Terrace} & \textbf{Aspen} & \textbf{CPAM} &
   \textbf{RadixGraph} & \textbf{Terrace} & \textbf{Aspen} & \textbf{CPAM}\\
\hline
 & $10$ & \textbf{\underline{6.58E5}} & 1.42E5 & 1.03E5 & 8.54E4
    & \textbf{\underline{4.89E5}} & 8.56E4 & 1.21E5 & 8.92E4
    & \textbf{\underline{1.09E5}} & 3.15E4 & 9.54E4 & 8.92E4\\
\textbf{Insertion} & $10^2$ & \textbf{\underline{4.47E6}} & 2.79E5 & 6.37E5 & 4.17E5
    & \textbf{\underline{7.23E6}} & 2.18E5 & 9.97E5 & 4.48E5
    & \textbf{\underline{9.34E5}} & 2.80E5 & \textbf{\underline{9.34E5}} & 4.86E5\\
\textbf{throughput} & $10^3$ & \textbf{\underline{3.36E7}} & 1.61E6 & 3.22E6 & 8.55E5
    & \textbf{\underline{3.38E7}} & 5.68E5 & 2.80E6 & 7.60E5
    & 1.32E6 & 9.87E5 & \textbf{\underline{3.45E6}} & 6.10E5\\
 & $10^4$ & \textbf{\underline{6.18E7}} & 2.74E6 & 6.10E6 & 2.61E6
    & \textbf{\underline{5.69E7}} & 2.67E6 & 6.57E6 & 1.99E6
    & \textbf{\underline{3.64E6}} & 2.45E6 & \textit{Seg fault} & 6.60E5\\
\hline
 & $10$ & \textbf{\underline{2.75E5}} & 2.29E5 & 9.57E4 & 8.32E4
    & \textbf{\underline{5.07E5}} & 2.55E5 & 1.27E5 & 9.10E4
    & \textbf{\underline{1.18E5}} & 1.09E5 & 1.02E5 & 8.85E4\\
\textbf{Deletion} & $10^2$ & \textbf{\underline{5.04E6}} & 1.01E6 & 6.09E5 & 4.33E5
    & \textbf{\underline{8.08E6}} & 6.10E5 & 9.55E5 & 4.61E5
    & \textbf{\underline{2.22E6}} & 7.20E5 & 9.77E5 & 5.29E5\\
\textbf{throughput} & $10^3$ & \textbf{\underline{6.10E7}} & 1.48E6 & 3.43E6 & 8.52E5
    & \textbf{\underline{7.07E7}} & 9.18E5 & 2.86E6 & 7.79E5
    & \textbf{\underline{1.96E7}} & 1.91E6 & 3.15E6 & 6.19E5\\
 & $10^4$ & \textbf{\underline{2.01E8}} & 2.81E6 & 6.92E6 & 2.71E6
    & \textbf{\underline{5.76E7}} & 1.41E6 & 8.07E6 & 2.09E6
    & \textbf{\underline{9.65E7}} & 5.44E6 & \textit{Seg fault} & 6.67E5\\
\hline
\textbf{Memory} & / & \textbf{\underline{0.77G}} & 1.51G & 3.52G & 2.57G &
\textbf{\underline{2.93G}} & 4.30G & 4.64G & 6.98G &
\textbf{\underline{26.25G}} & 47.06G & 66.78G & 69.60G\\
\hline
\end{tabular}
\caption{Throughput for inserting and deleting edges with varying batch sizes in the LJ, Orkut and Twitter graphs and memory consumptions of different methods.}
\label{tab:throughput}
\vspace{-2mm}
\end{table*}
\noindent\textbf{Swapping SORT with ART.}
To further compare SORT and ART, we integrate unodb (the ART implementation) as an alternative vertex index for RadixGraph. Then we evaluate and compare RadixGraph using ART and SORT to assess their relative performance. Table \ref{tab:ablation} presents the results and we summarize several insights:
\begin{itemize}[leftmargin=*]
\item \textbf{ART is memory-efficient and competitive with SORT.}
Across all datasets, the ART-based RadixGraph consumes slightly more memory than the SORT-based version. The four adpative node types: \textit{Node4}, \textit{Node16}, \textit{Node48}, and \textit{Node256}, corresponding to fixed child array sizes of 4, 16, 48, and 256, allow ART to reduce wasted space. However, the coarse granularity can still lead to memory overhead. For instance, a node with 49 children must upgrade to a \textit{Node256}, allocating space for 256 entries while storing only 49, thereby wasting memory. Although finer granularity (e.g., introducing more node types) could mitigate this issue, it would also increase transformation overhead due to more frequent node-type conversions. Nevertheless, on most datasets, the memory usage of the ART-based RadixGraph remains close to that of the SORT-based version, indicating that ART’s space efficiency is still competitive overall.
\item \textbf{Query costs of ART are higher than SORT, resulting in lower throughput.}
The ART-based RadixGraph exhibits significantly higher insertion and deletion times than the SORT-based version across all datasets (Table~\ref{tab:ablation}), indicating that both insertion and query operations in ART incur overhead. During insertions, ART often needs to resize its internal arrays when a node changes type (e.g., from \textit{Node4} to \textit{Node16}), which involves memory reallocation and data copying. For queries, ART performs a sequential scan (for \textit{Node4} and \textit{Node16}) within each node to locate the corresponding child pointer. These factors collectively lead to slower update and query performance compared with SORT. However, the overall slowdown ratio of ART is much more than the ``Insertion'' part of Table \ref{tab:sort-art}, and closer to the ``Query`` part. The reason is that in graph systems, edge updates occur much more frequently than vertex insertions. Therefore, a query-efficient vertex index can significantly improve the overall throughput.
\end{itemize}

\vspace{1mm}
\noindent\textbf{Graph analytics without edge chain.}
We evaluate the impact of disabling the edge chain structure in RadixGraph to assess its contribution to graph analytics performance. For 1-hop queries as well as PR, WCC, and LCC, the performance difference is negligible, as these tasks are largely insensitive to vertex-index lookup costs. For other analytical workloads, the observed slowdowns are summarized in Table~\ref{tab:ablation}. The results show that the edge chain generally improves performance.

\subsection{Batch updates}
\label{sec:compare-terrace}
Another line of edge storage research focuses on optimizing batch operations~\cite{aspen,cpam,terrace}, whereas RadixGraph is primarily designed for single-edge operations. In this subsection, we compare RadixGraph with these systems in terms of end-to-end efficiency and memory consumption under varying batch sizes to highlight the advantages of RadixGraph’s snapshot-log edge structure.

Table~\ref{tab:throughput} summarizes the results. We observe that RadixGraph scales effectively as the batch size increases, even though RadixGraph is not  explicitly optimized for batch processing. This scalability arises because RadixGraph’s $O(1)$ edge operation cost is independent of vertex degree, and larger batch sizes naturally improve cache locality. On the Twitter dataset, Aspen exhibits competitive performance for batch sizes $10^2$ and $10^3$, but fails with a segmentation fault at batch size $10^4$. In contrast, RadixGraph  outperforms other systems over most datasets and batch sizes. RadixGraph also demonstrates lowest memory consumption towards all datasets.

We note that single-edge and batch operations correspond to different real-world scenarios: single-edge operations are common in streaming or highly dynamic graphs, while batch operations are typical in bulk graph updates or offline analytics.

\section{Related works}
\label{sec:related}
\noindent\textbf{Indices for vertices and edges.} Various in-memory graph systems have been developed for dynamic graph storage and processing \cite{aspen,cpam,risgraph,sortledton,graphone,10.5555/2387880.2387884,teseo,llama,terrace,pma,livegraph,spruce,gastcoco,stinger}. Different indexing techniques have been explored for both vertex and edge management. Hash tables provide efficient lookups by mapping a large key space to a smaller one \cite{cuckoo,horton,10.14778/2850583.2850585}. Coupled with hashmaps, the multi-level vector stores vertices based on their logical IDs from 0 to $n-1$ \cite{sortledton,gtx}. B+ trees support logarithmic time operations \cite{bftree,10.14778/2752939.2752947,10.14778/2095686.2095688} and balanced binary search trees (BSTs) maintain elements in order and can be viewed as the binary equivalent of B-trees \cite{10.1145/1734714.1734731,10.1007/s00236-023-00452-6,10.1145/2689412}. Radix trees offer near-constant time operations and eliminate the need for complex splitting and merging \cite{trie,PATRICIA,ART}. Recent works have further optimized radix-based structures for dynamic graphs. Teseo \cite{teseo} employs an Adaptive Radix Tree (ART) \cite{ART} as its primary index, while Spruce \cite{spruce} adopts the van Emde Boas tree (vEB-tree) \cite{vEB} for vertex indexing.

\vspace{1mm}
\noindent\textbf{Graph databases.} Beyond in-memory dynamic graph systems, another approach focuses on developing graph databases that leverage external memory for storage. Neo4j \cite{neo4j2024} and OrientDB \cite{orientdb} use linked list-based storage, where each vertex maintains its adjacency list as a linked list. SQLG \cite{sqlg} adopts a relational model, storing vertices and edges as tables within a relational database. More recent efforts, such as Aster \cite{aster} and LSMGraph \cite{lsmgraph}, integrate Log-Structured Merge (LSM) trees \cite{lsm} with graph data structures to achieve high write throughput while maintaining efficient query performance. An industrial database, BG3 \cite{bg3}, adopt the Bw-tree \cite{bwtree} instead of LSM-tree to utilize cheap cloud storage and minimize costs. These databases are designed for persistent storage and can serve as backends for in-memory dynamic graph systems.
\section{Conclusion}
We presented RadixGraph, a fast and space-efficient data structure for in-memory dynamic graph systems. RadixGraph combines a space-optimized vertex index (SORT) with a compact edge structure that supports $O(1)$ dynamic edge operations while ensuring fast query performance. A potential direction for future work includes extending RadixGraph for transactional support, whose importance has been noted in GTX \cite{gtx}. While RadixGraph already supports multi-version concurrency control (MVCC), how to provide stronger transactional guarantees on top of the existing design remains an open problem. In particular, supporting stronger isolation levels would require coordinating the visibility and ordering of updates that span multiple vertex-local edge arrays (for example, a transaction involving multiple operations). This requires shared commit timestamps, detecting write-write conflicts, and other mechanisms to ensure consistent and correct transaction execution.
\section{Acknowledgement}
This research is supported by the National Research Foundation, Singapore under its Frontier CRP Grant (NRF-F-CRP-2024-0005), and NTU SUG-NAP (022029-00001). Any opinions, findings, and conclusions or recommendations expressed in this material are those of the author(s) and do not reflect the views of the National Research Foundation, Singapore.

\bibliographystyle{ACM-Reference-Format}
\balance
\bibliography{reference}
\clearpage
\appendix
\section{Proofs of Lemma \ref{lemma:optimality} and Theorem \ref{theorem:optimality}}
\textbf{Proof of Lemma \ref{lemma:optimality}.} Since $s_{l-1}\geq x$, we assume there exists a feasible solution $s_0',s_1',\cdots,s_{l-1}'$, where $s_{l-1}'>x$ such that the objective function is minimized. If $s'_{l-2}\leq x$, it is not hard to see by replacing $s'_{l-1}$ with $x$, then $(s_0',s_1',\cdots,x)$ is also a feasible solution such that:
$$f(s_0',s_1',\cdots,x)\leq f(s_0',s_1',\cdots,s_{l-1}')$$

If $s'_{l-2}>x$, we can verify $s'_{l-3}$ via the same process, until some $s_i'\leq x$. Then it is not hard to see:
$$f(s_0',s_1',\cdots,s_i',x,\cdots,x)\leq f(s_0',s_1',\cdots,s_{l-1}')$$

Therefore, we can always find an optimal solution such that $s_{l-1}^*=x$.

\noindent\textbf{Proof of Theorem \ref{theorem:optimality}.} Since $g(i,j)$ represents the optimal summation of the first $i$ terms in $f$ when $s_i=j$, $g(l-1,x)$ represents the optimal value of $f$ when $s_{l-1}=x$. By applying Lemma \ref{lemma:optimality}, $g(l-1,x)$ gives an optimal value of $f$.

\section{Proof of Theorem \ref{theorem:edge-complexity}}
Let $s_l$ be the size of the snapshot segment. Since the snapshot segment size equals the log segment size, the whole edge array is of length $s=2s_l$. Therefore, Algorithm \ref{alg:log-compaction} costs $O(s)=O(s_l)$ time since each bitmap operation costs $O(1)$ time.

W.L.O.G., now assume there are $m$ edge operations with $c$ compactions, where the size of the snapshot segment changes from $S_0,S_1,S_2,\cdots, S_c$. Denote $INS_i,DEL_i,UPD_i$ as the number of insertion, deletion and update operations during the $i$-th and $(i+1)$-th compactions. Since the log segment is of the same size of the snapshot segment, there are $INS_i+DEL_i+UPD_i=S_i$ operations between the $i$-th and $(i+1)$-th compaction. Since the $i$-th compaction costs $O(S_{i-1})$ time, the total time complexity of compactions is $O\left(\sum_{i=1}^cS_{i-1}\right)=O(\sum_{i=0}^{c-1}(INS_i+UPD_i+DEL_i))=O(m)$. Therefore, the amortized time complexity for each operation is $O(1)$.

\section{Worst-case analysis of SORT}
Our main analysis guarantees the optimality of the \emph{average} space complexity of SORT.
We now refine the worst-case analysis and show that, when vertex IDs are not excessively sparse,
the \emph{worst-case} space complexity of SORT is near \(O(n)\).

Let \(V=\{v_1,v_2,\dots,v_n\}\) be the set of vertices with strictly increasing IDs \(v_i<v_{i+1}\).
Define the bit-length of the largest identifier by $x=\lceil lg(v_n)\rceil$
and the maximum adjacent-ID gap by $g=\max_{1<i\le n}(v_i-v_{i-1})$. We analyze the space usage layer by layer. Recall that an internal node at layer \(i\) (for \(0\le i<l\))
contains an array of \(2^{a_i}\) child pointers. Layers with \(a_i=0\) contain no pointer arrays and can be safely pruned,
so in what follows we consider only layers with \(a_i\ge 1\).

The root (layer \(0\)) contributes \(O(2^{a_0})\) space for its pointer array.
Fix an index \(i\ge 1\). Each node at layer \(i\) represents a contiguous ID interval of size $S_i= 2^{a_i+a_{i+1}+\cdots+a_{l-1}}$. If such a node exists, then there is at least one actual vertex whose ID is in this interval. As adjacent IDs are separated by at most \(g\), each occupied interval of length \(S_i\) must contain at least $n_i=\max\{1,\; S_i/g\}$ vertices. Hence, the number of created nodes at layer \(i\) is at most \(n/n_i\), and the total space
for layer \(i\) is bounded by
\begin{equation}\label{eq:layer-i}
O\left(2^{a_i}\cdot\frac{n}{n_i}\right)=O\left(2^{a_i}\cdot n\cdot\min\left\{1,\frac{g}{S_i}\right\}\right)
\end{equation}

We split the layers into two types:
\begin{itemize}[leftmargin=*]
  \item \emph{Shallow layers} where \(S_i\ge g\). For these layers \(n_i=S_i/g\), and \eqref{eq:layer-i} becomes $O\left(2^{a_i}\cdot n\cdot\frac{g}{S_i}\right)=O\left(\frac{n g}{2^{a_{i+1}+\cdots+a_{l-1}}}\right)$.
  \item \emph{Deep layers} where \(S_i<g\). For each such layer the contribution is at most \(O(n\cdot 2^{a_i})\).
\end{itemize}

It is straightforward to check that there exists a partition index $j$, such that layers $i$ for $i<j$ are shallow layers and $i\geq j$ are deep layers, respectively.

Consider the contributions from the shallow layers. Since $a_k\geq 1$ for any $k$, as \(i\) increases the quantity
\(a_{i+1}+a_{i+2}+\cdots+a_{l-1}\) decreases by at least \(1\) at each step, so the denominators
\(2^{a_{i+1}+\cdots+a_{l-1}}\) form at least a geometric progression with ratio at least \(1/2\).
Consequently the deep-layer contributions are:
\[
\sum_{\text{shallow } i} \frac{n g}{2^{a_{i+1}+\cdots+a_{l-1}}}= O\left(ng\cdot\alpha\right),
\]
where $\alpha=\sum_{\text{shallow }i}\frac{1}{2^{a_{i+1}+\cdots+a_{l-1}}}\ll 1$ is the sum of the geometric series. It can only converge to $1$ when $l\to\infty$ and $a_{0}=a_1=\cdots=a_{l-1}=1$ but otherwise is far less than 1.

For deep layers, since at the partition index $j$ we have $S_j<g$, their contributions are:
$$\sum_{\text{deep } i}n\cdot 2^{a_i}\leq n\cdot S_j<ng=O(ng)$$
Note that $\sum_{\text{deep } i}2^{a_i}\leq \prod_{\text{deep }i}2^{a_i}=S_j$ since $a_i\geq 1$. The proof is straightforward by induction and therefore omitted.

Putting everything together, the total space  of the SORT satisfies:
$$\text{Space}=O(2^{a_0}+ng\cdot\alpha+ng)$$

When the ID gap \(g\) is small or is a constant (the typical case when identifiers are dense or roughly
uniformly spaced), the space complexity is bounded by:
$$\text{Space}=O(2^{a_0}+n)$$

Since $v_n<ng$, we have $a_0\leq x=\lceil lg(v_n)\rceil=O(lg(ng))$ and therefore $2^{a_0}\leq ng$. However, this worst-case bound is achieved only when $a_0\approx x$ but $a_0$ is usually much less than $x$. The case $a_0\approx x$ occurs only when vertex IDs are excessively dense (e.g., $g=2$ or $g=3$) or $l=1$. In that case, $g$ can also be viewed as a constant and $\text{Space}=O(n)$.

\section{An illustration of the optimization process}
\begin{figure}[t]
    \includegraphics[width=0.45\textwidth]{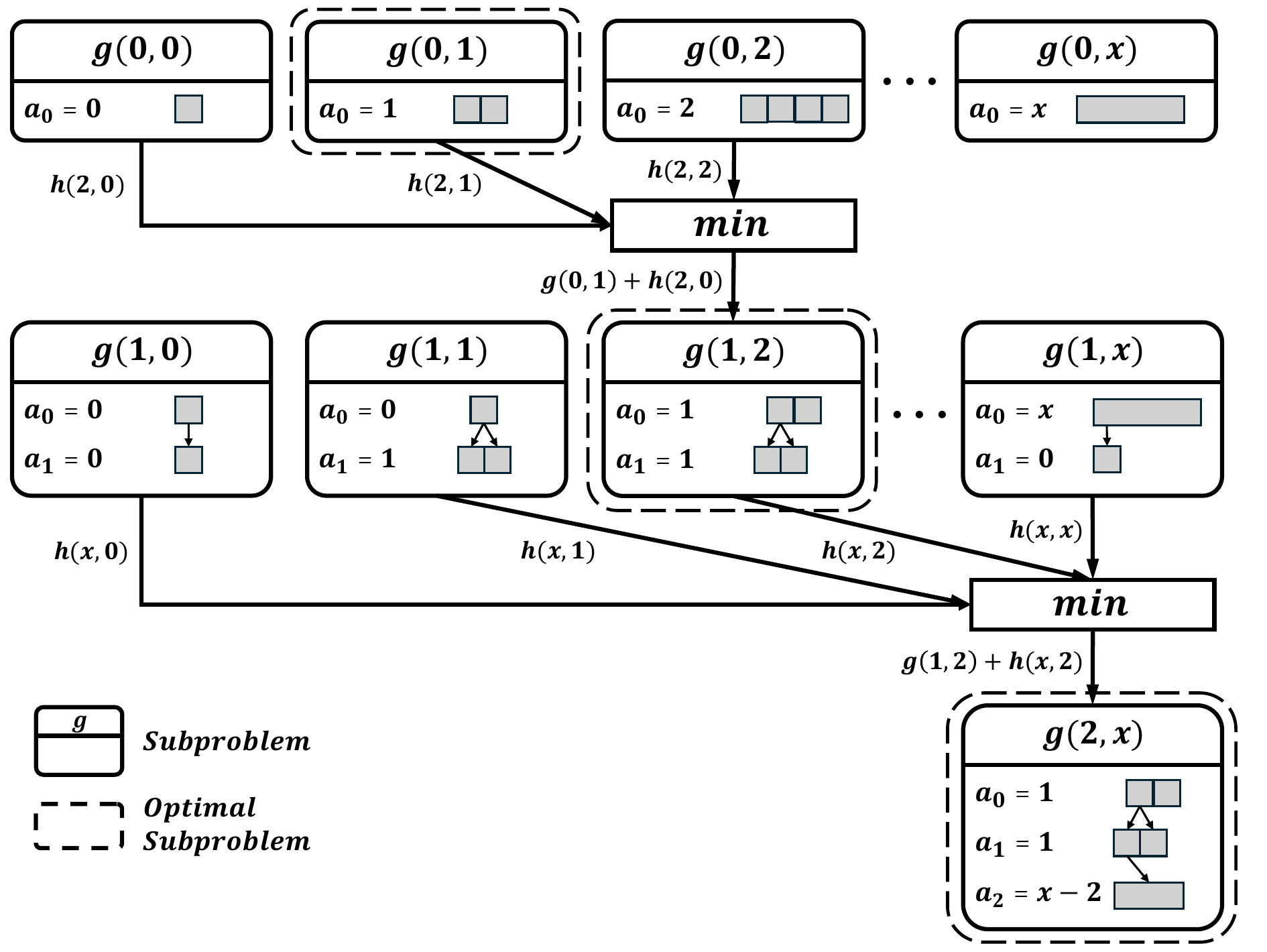}
    \caption{An example dynamic programming process when $l=3$. Here $h(j,k)=2^j\left(1-\frac{(2^x-n)!(2^x-2^{x-k})!}{(2^x)!(2^x-2^{x-k}-n)!}\right)$ is the transition cost from $g(i-1,k)$ to $g(i,j)$ in Equation \ref{eq:transition}.}
    \label{fig:dp}
\end{figure}
Figure \ref{fig:dp} shows an example dynamic programming process when 
$l=2$. By Theorem \ref{theorem:optimality}, the optimal solution is $s_0^*=1,s_1^*=2$ and $s_2^*=x$, which corresponds to an optimized radix tree setting: $a_0^*=1,a_1^*=1,a_2^*=x-2$.

\section{Time and space complexities of the optimization model}
Following Section \ref{sec:optimize-SORT}, we provide an analysis of the time and space consumptions to compute the optimal scheme of $g(l-1,x)$.
\begin{theorem}
    \label{theorem:time-complexity-sort}
    The time complexity of computing $g(l-1,x)$ is $O(nlx^2)$ and the space complexity is $O(lx)$.
\end{theorem}
\begin{proof}
To compute $g(l-1,x)$, we need to compute $g(i,j)$ for all $0\leq i<l-1,0\leq j\leq x$, which contains $O(lx)$ subproblems. To compute each $g(i,j)$, we need to enumerate $k$ for all $0\leq k\leq j$ (i.e., $O(x)$ subproblems) to find the optimal transition. Each transition requires $O(n)$ multiplications and divisions. Therefore, the total time complexity is $O(nlx^2)$. For the space complexity, we need to store $g(i,j)$ for all $O(lx)$ subproblems and the previous optimal subproblem $g(i-1,k)$ contributing to $g(i,j)$, such that we can find a path for optimal transition to $g(l-1,x)$ which recovers the optimal $a_i$.
\end{proof}

Practically, we can further apply pruning to speed up the computation without harming the results. Since $\frac{p}{q}<\frac{p+1}{q+1}$ when $p<q$, we know that the factorial part is not larger than $\left(\frac{2^x-2^{x-k}}{2^x}\right)^n$. Therefore, if $g(i-1,k)+2^j\left(1-\left(\frac{2^x-2^{x-k}}{2^x}\right)^n\right)$ is larger than current $g(i,j)$, we can safely prune $g(i-1,k)$ without computing factorials.
\end{document}